\documentclass[11pt, a4paper]{article}
\pdfoutput=1 %

\usepackage{fullpage}

\usepackage{mdframed}

\usepackage[utf8]{inputenc}

\usepackage{libertine}

\usepackage{amsfonts}
\usepackage{amsmath}
\usepackage{amssymb}
\usepackage{amsthm}
\usepackage{float}
\usepackage{stmaryrd}
\usepackage{mathtools}

\usepackage{algorithm}
\usepackage{algorithmic}

\usepackage{hyperref}
\usepackage[svgnames]{xcolor}
\hypersetup{colorlinks={true},urlcolor={blue},linkcolor={DarkBlue},citecolor=[named]{DarkGreen}}
\usepackage[numbers,sort&compress]{natbib}

\usepackage{newfloat}
\usepackage{listings}
\floatstyle{ruled}
\newfloat{listing}{tb}{lst}{}
\floatname{listing}{Listing}

\usepackage{xspace}
\usepackage{fancyhdr}
\usepackage{tcolorbox}

\usepackage{microtype}
\usepackage[capitalise,nameinlink]{cleveref}
\usepackage{enumitem}

\usepackage{cleveref}

\usepackage{doi}

\usepackage{authblk}

\usepackage{booktabs} %

\usepackage{tikz}  %
\usetikzlibrary{arrows}
\usetikzlibrary{patterns}
\usetikzlibrary{decorations.shapes}
\usetikzlibrary{backgrounds}
\usetikzlibrary{calc}
\tikzstyle{overbrace text style}=[font=\tiny, above, pos=.5, yshift=5pt]
\tikzstyle{overbrace style}=[decorate,decoration={brace,raise=5pt,amplitude=3pt}]
\usetikzlibrary{shapes.geometric}
\usetikzlibrary{hobby}

\usepackage{authblk}
\usepackage{xspace}

\usepackage{xspace}
\usepackage{amsmath}
\usepackage{amsthm}
\usepackage{amsfonts}
\usepackage{amssymb}
\usepackage{cleveref}
\usepackage{tikz}
\usetikzlibrary{calc,backgrounds}

\newcommand{\agents}{\ensuremath{A}} %
\newcommand{\parts}{\ensuremath{k}} %
\newcommand{\wFn}{\ensuremath{\omega}} %
\newcommand{\util}{\ensuremath{\operatorname{u}}} %
\newcommand{\Fr}{\ensuremath{F}} %
\newcommand{\pttn}{\ensuremath{\pi}} %

\newcommand{\probName}[1]{\textsc{#1}\xspace}

\renewcommand{\P}{\textsf{P}\xspace}
\newcommand{\NP}{\textsf{NP}\xspace}

\newcommand{\NPhness}{\NP-hardness\xspace}
\newcommand{\NPc}{\NP-complete\xspace}

\newcommand{\FPT}{\textsf{FPT}\xspace}
\newcommand{\XP}{\textsf{XP}\xspace}
\newcommand{\W}[1][1]{\textsf{W[#1]}\xspace}
\newcommand{\Wh}[1][1]{\W[#1]-hard\xspace}
\newcommand{\Whness}[1][1]{\W[#1]-hardness\xspace}

\newcommand{\Oh}[1]{\ensuremath{{\mathcal{O}(#1)}}}

\newcommand{\tw}{\operatorname{tw}}
\newcommand{\vc}{\operatorname{vc}}

\newcommand{\N}{\ensuremath{\mathbb{N}}}

\newtheorem{theorem}{Theorem}
\newtheorem{proposition}{Proposition}
\newtheorem{observation}{Observation}
\newtheorem{corollary}{Corollary}
\newtheorem{definition}{Definition}

\newtheorem{claim}{Claim}
\crefname{claim}{Claim}{Claims}
\newenvironment{claimproof}[1]{\par\noindent\emph{Proof.}\hspace{0.15cm}#1}{\hfill~$\blacktriangleleft$\smallskip}

\newcommand{\pgfextractangle}[3]{%
\pgfmathanglebetweenpoints{\pgfpointanchor{#2}{center}}
{\pgfpointanchor{#3}{center}}
\global\let#1\pgfmathresult
}

\newcommand{\convexpath}[2]{
[
create hullnodes/.code={
	\global\edef\namelist{#1}
	\foreach [count=\counter] \nodename in \namelist {
		\global\edef\numberofnodes{\counter}
		\node at (\nodename) [draw=none,name=hullnode\counter] {};
	}
	\node at (hullnode\numberofnodes) [name=hullnode0,draw=none] {};
	\pgfmathtruncatemacro\lastnumber{\numberofnodes+1}
	\node at (hullnode1) [name=hullnode\lastnumber,draw=none] {};
},
create hullnodes
]
($(hullnode1)!#2!-90:(hullnode0)$)
\foreach [
evaluate=\currentnode as \previousnode using \currentnode-1,
evaluate=\currentnode as \nextnode using \currentnode+1
] \currentnode in {1,...,\numberofnodes} {
	let
	\p1 = ($(hullnode\currentnode)!#2!-90:(hullnode\previousnode)$),
	\p2 = ($(hullnode\currentnode)!#2!90:(hullnode\nextnode)$),
	\p3 = ($(\p1) - (hullnode\currentnode)$),
	\n1 = {atan2(\y3,\x3)},
	\p4 = ($(\p2) - (hullnode\currentnode)$),
	\n2 = {atan2(\y4,\x4)},
	\n{delta} = {-Mod(\n1-\n2,360)}
	in
	{-- (\p1) arc[start angle=\n1, delta angle=\n{delta}, radius=#2] -- (\p2)}
}
-- cycle
}

\title{Balanced and Fair Partitioning of Friends%
    \footnote{%
        An extended abstract of this work has been published in the Proceedings of the 39th {AAAI} Conference on Artificial Intelligence, AAAI~'25~\cite{DeligkasEIKS2025}. %
    }%
}

\author[1]{Argyrios Deligkas}
\author[1]{Eduard Eiben}
\author[1]{Stavros D. Ioannidis}
\author[2]{Dušan Knop}
\author[2]{Šimon~Schierreich}

\affil[1]{Royal Holloway, University of London, United Kingdom}
\affil[2]{Czech Technical University in Prague, Czech Republic}

\begin{document}

\maketitle

\begin{abstract}
    In the recently introduced model of {\em fair partitioning of friends}, there is a set of agents located on the vertices of an underlying graph that indicates the friendships between the agents. 
    The task is to partition the graph into~$k$ balanced-sized groups, keeping in mind that the value of an agent for a group equals the number of edges they have in that group.
    The goal is to construct partitions that are ``{\em fair}'', i.e., no agent would like to replace an agent in a different group.
    We generalize the standard model by considering utilities for the agents that are beyond binary and additive. 
    Having this as our foundation, our contribution is threefold
    (a) we adapt several fairness notions that have been developed in the fair division literature to our setting;
    (b) we give several existence guarantees supported by polynomial-time algorithms;
    (c) we initiate the study of the computational (and parameterized) complexity of the model and provide an almost complete landscape of the (in)tractability frontier for our fairness concepts.
\end{abstract}

\section{Introduction}
\label{sec:intro}

You are the coordinator of your organization's annual banquet, and your task is to allocate seats on tables for the employees. 
As you aim for perfection, you want to ensure that every employee believes that they are part of one of the ``best'' tables of the banquet.
In other words, you want each employee to value the company of their table ``almost'' as much as the company they would get if they replace an employee at a different table.
However, you have the following constraints: (a) there are only~$k$ tables of the same size that exactly fit the participants; (b) friendships between employees vary greatly; (c) you want to be ``fair'' to every employee. 
So, what seat allocation would bring all employees to the tables?

Scenarios like the above appear in several other domains where the goal is to partition a group of people into~$k$ groups of almost the same size: students for team projects, employees to training groups, %
kids for camp houses, or clustering in general.
Recently, \citet{LiMNS2023} introduced an elegant framework in order to formally study those situations. 
In their model, there is a {\em friendship graph} where every node corresponds to an agent and every edge corresponds to a friendship between the agents. 
The task is to create~$\parts$ groups of almost equal size such that the resulting grouping is ``fair''. They have focused on two fairness notions: a version of the {\em core}, refined for this model, and {\em almost envy-free} partitions, denoted EF$r$, which is most relevant to this work. 
In an EF$r$ partition, no agent could increase the number of friends by at least~$r$ by replacing an agent from another group.
\citeauthor{LiMNS2023} showed that for general friendship graphs and~$r=\Oh{\sqrt{\frac{n}{\parts} \log \parts}}$, where~$n$ is the number of agents and~$\parts$ is the number of parts, EF$r$ partitions always exist and can be computed in polynomial time. 
In addition, they showed that EF1 partitions do not always exist, even when~$\parts=2$.

Although the model of \citeauthor{LiMNS2023} is elegant and sets the foundation of the framework, it implicitly assumes that: 
(a) every agent values their friends equally, i.e., every friend counts the same, thus friendships are symmetric;
(b) the value over friends is additive.
Furthermore, the computational complexity of finding fair partitions remained unexplored.
As one can imagine, the aforementioned assumptions can be rather restrictive for a variety of real-life situations.
Can we augment the model in order to handle a richer variety of situations?
In addition, the factor~$r=\Oh{\sqrt{\frac{n}{k} \log k}}$ for EF$r$ in general graphs could be prohibitive to convince the agents that the partition is fair.
Having said this, the space between general graphs and more restricted graph classes is vast. Can one exploit the structure of the friendship graph and provide stronger fairness guarantees?
The goal of the paper is to shed some light on the above-mentioned questions.

\subsection{Our Contribution}
Our initial contribution is the generalization of the model of \citeauthor{LiMNS2023}~\cite{LiMNS2023}, which goes beyond binary, symmetric and additive utilities, and thus can capture more real-life scenarios. Having this as our foundation, our contribution is threefold:
\begin{enumerate}[label=(\alph*)]
    \item we adapt several fairness notions that have been developed in the fair division literature to our setting; that allows for the choice of the most appropriate one depending on the scenario in which our model is applied;
    \item we give several existence guarantees supported by polynomial-time algorithms;
    \item we initiate the study of the computational (and parameterized) complexity of the model and provide an almost complete landscape of the (in)tractability frontier for our fairness concepts with respect to various restrictions of the problem.
\end{enumerate}

More specifically, in~\Cref{sec:fairness}, we define the fairness notions for our model (namely, EF,~$\text{EFX}_0$, EFX, EF1, PROP, and MMS), and we provide a complete taxonomy for their interconnections. Next, in~\Cref{sec:complexity}, we provide a thorough study of the computational complexity of deciding the existence of fair allocations. Unfortunately, the main takeaway message there is negative: all fairness notions are intractable, even in very simple cases, such as if the friendship graph is a path (\Cref{thm:PROP:NPh}), or is bipartite, with vertex cover of size 2, and we aim only for two parts (\Cref{thm:EFX:NPh,thm:EF1_NPh_vc2,thm:MMS:nopolytime}). On the positive side, we show that the problem becomes {\em fixed-parameter tractable}, for {\em all} fairness notions of our interest, if all friendships have the same value and we parameterize by the vertex cover number (\Cref{thm:binaryFPTbyVC}). Note that the vertex cover number is a very natural parameter for our problem. In many scenarios similar to our initial motivational story, there is a small set of super-stars/influencers/politicians, and many people who do not know each other want to be close to them.

The bulk of our positive results is given in \Cref{sec:trees}, which focuses on tree-like friendship graphs. The main message is that the existence of a fair allocation and its tractability essentially depends on the concept of fairness we adopt. In fact, we provide a strong dichotomy. For half of the fairness notions we consider (MMS, EFX, and EF1), we prove that they {\em always} admit a solution, even if the utilities are monotone. Actually, we provide algorithms that compute such allocations that become polynomial-time if the utilities are additive (\Cref{thm:monotone:MMS:exists}). However, the problem is intractable for the remaining notions even on trees with binary utilities (\Cref{thm:unweighted:EF:tree:NPh}). On the positive side, we complement our negative results with ``semi''-efficient algorithm: when the number of parts~$\parts$ is constant, utilities are binary, and the friendship graph is of constant tree-width, the problem can be solved in polynomial time (\Cref{thm:tw:XP}). The algorithm is based on dynamic programming and is the best possible given our previous hardness results.

Before we conclude in \Cref{sec:conclusion} with some interesting directions for future research, in \Cref{sec:discussion}, we extensively discuss how our hardness and algorithmic results generalize beyond balanced partitions. Specifically, we show that our results completely characterize the complexity landscape of fairness in the more general model of additively separable hedonic games with fixed-size coalitions, as introduced by \citeauthor{BiloMM2022}~\cite{BiloMM2022}.

\subsection{Related Work}

Our work, perhaps surprisingly, builds upon the foundations of three different subfields of AI research: \emph{fair division}, \emph{coalition formation}, and \emph{clustering}.

In the \emph{fair division of indivisible items} model~\cite{BramsT1996,BouveretCM2016,AmanatidisABFLMVW2023}, we are usually given a set of indivisible items and a set of agents together with their preferences over the items, and the goal is to find an allocation of items to the agents that is fair with respect to some well-defined notion of fairness. The crucial difference from our model is that, in our case, the set of agents and the set of items coincide. Additionally, we are given a friendship graph that further restricts the agents' preferences. Few papers have been published on fair division in the presence of an additional social network~\cite{AzizBCGL2018,BeynierCGHLMW2019,BredereckKN2022,EibenGHO2023}. However, in these works, the social network restricts the communication between agents and does not encode their preferences as in our case. A different line of work on fair division also considered the presence of graphs~\cite{BouveretCEIP2017,ChristodoulouFKS2023,ZhouWLL2024,MisraS2024}; however, here, the edges correspond to items and do not encode the preferences.

\emph{Additively-Separable Hedonic Games} (ASHGs)~\cite{BogomolnaiaJ2002} is an important special case of \emph{hedonic games}~\cite{DrezeG1980,AzizS2016}, where we are given a set of agents, the utility each agent~$a$ receives from being in the same group as an agent~$b$, and the task is to partition the agents into several groups, called coalitions, such that the outcome is stable. There are also several works that study fairness in the context of coalition formation~\cite{AzizBS2013,WrightV2015,Peters2016a,Peters2016b,Ueda2018,BarrotY2019,KerkmannNR2021}, but none of these works requires coalitions of fixed size. Such settings, where either the number of coalitions or the coalitions' sizes are prescribed, were studied in recent works~\cite{BiloMM2022,SlessHKW2018,CsehFH2019,LevingerAH2023,MonacoM2023}. However, all of these papers suppose stability and not fairness as a solution concept. In this context, related are also \emph{friend and enemies games}~\cite{DimitrovBHS2006,OhtaBISY2017,BarrotOSY2019,FlamminiKV2022,SkibskiSGSMY2022,ChenCRS2023}, where the set of agents is given together with a friendship graph, where each edge represents that two agents are either friends or enemies. This is clearly different from our work, as we allow only friends, and, at the same time, the strength of friendships can be different for every pair of agents. Additionally, we require a fixed number of parts.

The third field where our results may find application is \emph{clustering}~\cite{Jain2010,EverittLLS2011} and specifically \emph{data microaggregation}~\cite{DomingoFerrer2009,BlazejGKPSS2023}. In related problems, the goal is usually to partition a set of data points into~$k$ \emph{clusters} such that the clusters have some desirable properties that vary based on a particular application. The closest problem to our work is proportionally fair clustering~\cite{ChenFLM2019,MichaS2020,caragiannis2024proportional}. However, in contrast to our work, they do not suppose an underlying graph, and they require the solution to consist additionally of~$k$ representatives/centroids of every cluster.

Finally, we would like to mention that balanced partitioning of graphs is also a widely studied problem in graph theory. Arguably, the closest problem to what we study in this work is the so-called \probName{Equitable Connected Partition} (ECP for short), where we are given a graph and a number of parts~$\parts$, and the goal is to partition the graph into~$\parts$ connected sub-graphs of almost equal sizes. This problem, heavily studied from the computational complexity perspective~\cite{GareyJ1979,Altman1997,EncisoFGKRS2009,BlazejKPS2024}, differs from ours in the sought solution concept; in ECP, all the parts need to be connected, while in our model, every solution partition needs to be fair, which does not imply connectedness. It should be pointed out that ECP has many applications, including redistricting theory~\cite{Williams1995,Altman1997,LandauRY2009,Schutzman2020,KoTAM2022,ChenBCLRL2023}, an important sub-field of computational social choice.

\section{Preliminaries}

By~$\N$, we denote the set of positive integers~$\{1,2,3,\ldots\}$. For every~$i\in\N$, we let~$[i] = \{1,\ldots,i\}$ and~$[i]_0 = \{0\}\cup[i]$.

The input of our problem consists of a set~${\agents = \{a_1,\ldots,a_n\}}$ of~$n$ \emph{agents}. The relations between the agents are represented by an undirected graph~$G=(\agents,E)$, called a \emph{friendship graph}, where each edge~$\{a,b\}\in E$ represents the friendship between agents~$a$ and~$b$. Observe that the set of vertices and the set of agents coincide and that the friendship is always mutual. By~$\Fr(a) = \{ b \mid \{a,b\}\in E\}$, we denote the set of \emph{friends} of an agent~$a\in\agents$. The number of friends, denoted by~$\deg(a) = |\Fr(a)|$, is called the \emph{degree} of agent~$a$. By~$\Delta$, we denote the maximum degree in our friendship graph, that is,~$\Delta = \max_{a\in\agents} \deg(a)$. For a graph~$G$, we denote by~$V(G)$ the set of its vertices and by~$E(G)$ the set of its edges. For additional graph-theoretical notation, we follow the monograph of \citet{Diestel2017}.

For every agent~$a\in \agents$, we have a \emph{utility function}~$\util_a\colon \agents\to\N\cup\{0\}$ such that, for an agent~$b\in A$, we have~$\util_a(b) > 0$ if~$b\in \Fr(a)$ and~$0$ otherwise. If~$\util_a(b) = \util_b(a)$ for every~$a,b\in \agents$, we say that the utilities are \emph{symmetric}. Additionally, if for every~$b\in\agents$ it holds that~$\util_a(b)$ is the same for every~$a\in \Fr(b)$, the utilities are called \emph{objective}, and we write~$\util_\forall(a) = x$ to denote that the utility of all~$\Fr(a)$ for~$a$ is the same~$x > 0$. If for some~$a\in\agents$ we have~$\util_a(b) = \util_a(b')$ for every pair of~$b,b'\in F(a)$, we say that the agent~$a$ is \emph{indistinguishing}. Finally, if~$\util_a(b) \in \{0,1\}$ for every~$a,b\in\agents$, we say that the utilities are \emph{binary}. Observe that binary utilities are symmetric and objective, and all agents are indistinguishing.

Our goal is to find a \emph{partition} of agents into~$\parts\leq n$ \emph{groups} of almost the same sizes, which is fair with respect to some well-defined notion of fairness. A \emph{$\parts$-partition} of~$\agents$ is a list~$\pttn = (\pttn_1,\ldots,\pttn_k)$ such that~$\bigcup_{i\in[k]} \pttn_i = \agents$,~$\pttn_i \cap \pttn_j = \emptyset$ for every pair of distinct~$i,j\in[k]$, and no~$\pttn_i$ is an empty set. Observe that a~$k$-partition always exists. If~$\parts$ is clear from the context, we refer to a~$\parts$-partition just as a partition. By~$\pttn(a)$, we denote the part in which agent~$a\in\agents$ is placed in~$\pttn$. 
A~$\parts$-partition~$\pttn$ is called \emph{balanced} if every pair of groups differs in their sizes by at most one. In other words, in a balanced partition we have~$\lfloor \frac{n}{\parts} \rfloor \leq |\pttn_i| \leq \lceil \frac{n}{\parts} \rceil$ for every~$i\in[k]$. Otherwise,~$\pttn$ is called \emph{imbalanced}. Unless stated explicitly, we assume only balanced partitions. By~$\Pi_\parts$, we denote the set of all (balanced)~$\parts$-partitions.

To keep the notation clear, we extend utility functions from single agents to groups of agents. More specifically, let~$S\subseteq\agents$ be a set of agents. Then, to denote the utility of an agent~$a\in\agents$ in set~$S$, we write~$\util_a(S)$. A utility function is \emph{additive}, if it holds that~$\util_a(S) = \sum_{b\in S} \util_a(b)$. 
In some parts of the paper, we also study monotone utilities. 
A utility function is called \emph{monotone}, if for every pair of sets~$S,T$ it holds that~$S \subset T \implies \util(S) \leq \util(T)$. Unless stated otherwise, we assume that the utilities under consideration are additive. Similarly, we set the utility of an agent~$a\in\agents$ in a partition~$\pttn$ to be~$\util_a(\pttn) = \util_a(\pttn(a))$.

\subsection{Parameterized Complexity}

Parameterized complexity studies the complexity of a problem with respect to its input size,~$n$,  and the size of a parameter~$k$. A problem is called {\em fixed-parameter tractable} with respect to a parameter~$k$ if it can be solved in time~$f(k)\cdot \operatorname{poly}(n)$, where~$f$ is a computable function. 
A less favorable, but still positive, outcome is a so-called \XP \emph{algorithm}, which has running-time~$n^\Oh{f(k)}$; problems admitting such algorithms belong to the complexity class~$\XP$.
Showing that a problem is \Wh rules out the existence of a fixed-parameter algorithm under the well-established assumption that~$\W\neq \FPT$. For a more comprehensive introduction to the framework of parameterized complexity, we refer the interested reader to the monograph of \citet{CyganFKLMPPS2015}.

\subsection{Structural Parameters} 
Many of our results indicate that the problem of our interest is not tractable in its full generality. Therefore, we study different restrictions of the underlying input friendship graph. The first structural parameter we use is the vertex cover number which, informally, is the minimum number of vertices we need to remove from~$G$ to obtain an edge-less friendship graph.

\begin{definition}[Vertex Cover Number]
	Let~$G=(\agents,E)$ be a friendship graph. A set~$C\subseteq \agents$ is a \emph{vertex cover} of~$G$ if for every~$\{u,v\}\in E$ we have~$u\in C$ or~$v\in C$. The \emph{vertex cover number} of~$G$, denoted by~$\vc(G)$, is the size of the smallest vertex cover of~$G$.
\end{definition}

In the second half of our paper, we are interested in tree-like friendship graphs. One structural parameter that captures how close a graph is to being a tree is the celebrated treewidth, defined below.

\begin{definition}[Treewidth]
	Let~$G=(\agents,E)$ be a friendship graph. A \emph{tree decomposition} of~$G$ is a triple~$(T,\beta,r)$, where~$T$ is a tree rooted in node~$r$ and~$\beta\colon V(T)\to 2^\agents$ is a mapping such that
	\begin{enumerate}
		\item every agent~$a\in\agents$ is mapped to at least one node~$x\in V(T)$, that is,~$\bigcup_{x\in V(T)} \beta(x) = \agents$,
		\item for every edge~$\{a,b\}\in E$, there exists a node~$x\in V(T)$ such that~$\{a,b\} \subseteq \beta(x)$, and
		\item for every agent~$a\in \agents$, the nodes~$\{ x \in V(T) \mid a\in \beta(x) \}$ for a connected sub-tree of~$T$.
	\end{enumerate}
	
	For a node~$x\in V(T)$, we refer to the set~$\beta(x)$ as a \emph{bag} of~$x$. The width of a tree decomposition is~${\max_{x\in V(T)} |\beta(x)| - 1}$ and the \emph{treewidth} of~$G$, denoted~$\tw(G)$, is the minimum width of a tree decomposition over all tree decompositions of~$G$.
\end{definition}

For algorithmic purposes, we will work with a slight variation of the basic tree decomposition.

\begin{definition}[Nice Tree Decomposition]
	A tree decomposition~$(T,\beta,r)$ is called \emph{nice} if the following conditions hold
	\begin{itemize}
		\item~$\beta(r) = \beta(\ell) = \emptyset$ for every leaf~$\ell\in V(T)$.
		\item every non-leaf node~$x\in V(T)$ is one of the following types:
		\begin{description}
			\item[Introduce Node:] a node with exactly one child~$y\in V(T)$ such that~$\beta(x) = \beta(y) \cup \{a\}$ for some agent~$a\not\in \beta(y)$,
			\item[Forget Node:] a node with exactly one child~$y\in V(T)$ such that~$\beta(x) = \beta(y) \setminus \{a\}$ for some agent~$a\in \beta(y)$, or
			\item[Join Node:] a node with exactly two children~$y,z\in V(T)$ such that~$\beta(x) = \beta(y) = \beta(z)$.
		\end{description}
	\end{itemize}
\end{definition}

It is known that, given a tree decomposition of width~$w$, one can compute an equivalent nice tree decomposition of the same width with~$\Oh{w\cdot|V(T)|}$ nodes in~$\Oh{w^2\cdot\max\{|V(T)|,|V(G)|\}}$ time~\cite{CyganFKLMPPS2015}.

\section{Fairness Concepts}\label{sec:fairness}

In this section, we formally define the studied notions of fairness and settle their basic properties and relationships. We mirror the respective definitions for the standard fair division of indivisible items~\cite{BramsT1996,BouveretCM2016,AmanatidisABFLMVW2023}.

The first notion of fairness is based on envy between pairs of agents and was already studied by \citet{LiMNS2023}.%

\begin{definition}[EF]
	A~$\parts$-partition~$\pttn$ is called \emph{envy-free (EF)} if, for every pair of agents~$a,b\in\agents$,~$\util_a(\pttn(a)) \geq \util_a(\pttn(b)\setminus\{b\})$.
\end{definition}

It is easy to see that envy-free partitions are not guaranteed to exist: assume an instance where the friendship graph is a path with three agents, the utilities are binary, and~$\parts=2$. If the smaller part consists of a degree-two agent, then this agent has envy towards any agent in the larger part. Similarly, if the smaller part consists of a degree-one agent, this agent is envious towards the other degree-one agent.

As illustrated by the previous example, envy-freeness is generally not guaranteed to exist, even in very simple instances. Therefore, we introduce several relaxations of EF, which are very popular and heavily studied in the standard fair division literature.

The first relaxation is inspired by the work of \citet{PlautR2020} and imposes that any envy between two agents can be eliminated by the removal of an arbitrary agent from the target part.

\begin{definition}[EFX$_0$]
	A~$\parts$-partition~$\pttn$ is called
 \emph{envy-free up to any agent (EFX$_0$)} if, for every pair of agents~$a,b\in\agents$ and every agent~$c\in\pttn(b)\setminus\{b\}$ it holds that 
~$\util_a(\pttn(a)) \geq \util_a(\pttn(b)\setminus\{b,c\})$.
\end{definition}

It is easy to observe that every EF partition is also EFX$_0$. To see that this is not the case in the opposite direction, recall the counterexample for the existence guarantee of EF; if we are allowed to remove any agent in the target part, we obtain that any 2-partition for this instance is EFX$_0$.

By the definition of EFX$_0$, the envy of~$a$ towards~$b$ should be eliminated by the removal of any \emph{agent}~$c$, even if this agent is not a friend of agent~$c$. That is,~$\util_a(c) = 0$. This may be counter-intuitive in some scenarios, and therefore, we also introduce a more restricted variant, which requires that the envy is eliminated by the removal of arbitrary \emph{friend}. This notion of fairness is in line with the envy-freeness up to any good, which is due to \citet{CaragiannisKMPSW19}.

\begin{definition}[EFX]
	A~$\parts$-partition~$\pttn$ is called \emph{envy-free up to any friend (EFX)} if, for every pair of agents~$a,b\in\agents$ and every agent~$c\in(\pttn(b)\setminus\{b\})\cap \Fr(a)$ it holds that 
   ~$\util_a(\pttn(a)) \geq \util_a(\pttn(b)\setminus\{b,c\})$.
\end{definition}

Both above-defined relaxations can also be understood such that the envy between two agents can be eliminated by removing a least-valued friend (or agent in the case of EFX$_0$) from the target part. We can give an even more relaxed envy-based notion that allows the removal of the {\em most valued} friend. This is an adaptation of EF1 from the standard fair division literature~\cite{LiptonMMS04,Budish11}.

\begin{definition}[EF1]
	A~$\parts$-partition~$\pttn$ is called \emph{envy-free up to one agent (EF1)} if, for every pair of agents~$a,b\in\agents$, there exists an agent~$c\in\pttn(b)\setminus\{b\}$ such that~$\util_a(\pttn(a)) \geq \util_a(\pttn(b)\setminus\{b,c\})$.
\end{definition}

It was already proved by \cite{LiMNS2023} that, even if the utilities are binary, there are instances with no EF1 partition for any~$\parts \geq 2$. Consequently, the same non-existence result also holds for EFX and EFX$_0$.

The fairness notions introduced so far were based on comparing the agent's current part with all other parts they can move to. The following set of fairness notions deviates from comparison with other agents and is based solely on certain threshold values that should be guaranteed for every agent in each fair outcome.

\begin{definition}[PROP]
	For every agent~$a\in\agents$, we set~$\operatorname{PROP-share}(a) = \util_a(\Fr(a))/\parts$. A~$\parts$-partition~$\pttn$ is called \emph{proportional (PROP)} if for every agent~$a\in\agents$ it holds that~$\util_a(\pttn(a)) \geq \operatorname{PROP-share}(a)$.
\end{definition}

It is again easy to see that proportional partitions are not guaranteed to exist: assume again the instance where the friendship graph is a path on three agents, utilities are binary, and~$\parts = 2$. The~$\operatorname{PROP-share}$ of the degree-two agent is exactly one, and therefore, this agent is necessarily in part with one of its friends. Consequently, the other friend is alone in the second part; the utility of this agent is clearly zero, but its~$\operatorname{PROP-share}$ is~$1/2$. Hence, no partition is proportional for this instance. 

A more relaxed fairness notion compared to PROP is the maxi-min share, denoted MMS.
There each agent performs a though experiment where their threshold is the maximum value they can get among all possible allocations, where in each allocation they gets allocated to the ``worst'' part of it.

\begin{definition}[MMS]
	For every agent~$a\in\agents$, we set~$\operatorname{MMS-share}(a) = \max_{\pttn'\in\pttn_\parts}\min_{i\in[k]} \util_a(\pttn'_i)$.
	A~$\parts$-partition~$\pttn$ is called \emph{maxi-min share (MMS)} if for every agent~$a\in\agents$ it holds that~$\util_a(\pttn(a)) \geq \operatorname{MMS-share(a)}$.
\end{definition}

In our first result, we show that MMS partitions are not guaranteed to exist for any number of parts~$\parts \geq 2$.

\begin{proposition}
\label{thm:MMS:nonexistence}
	For every~$\parts \geq 2$, an MMS partition is not guaranteed to exist, even if the utilities are binary.
\end{proposition}
\begin{proof}
	Let~$k\geq 2$ be fixed. We start by defining the friendship graph. First, we create~$\parts+1$ \emph{standard-agents}~$i_1,\ldots,i_{\parts+1}$. Next, we create~$\parts$ \emph{guard-agents}~$a_1,\ldots,a_\parts$, which are all friends with each other; that is, they form a clique. Finally, we connect each guard-agent with every standard-agent. Clearly, it holds that~$\Fr(i_j) = \Fr(i_\ell)$ for all~$j,\ell\in[\parts+1]$ and~$\Fr(a_j)\cup\{a_j\} = \Fr(a_\ell)\cup\{a_\ell\}$ for all~$j,\ell\in[\parts]$. Moreover, the MMS-share of every standard-agent is exactly~$\parts/\parts = 1$, and the MMS-share of every guard-agent is exactly~$(\parts-1+\parts+1)/\parts = 2\parts/\parts = 2$. Overall, we have~$2\parts + 1$ agents, and therefore, in every balanced partition, there are~$\parts-1$ parts of size~$2$ and one part of size~$3$. For the sake of contradiction, assume that there exists an MMS partition~$\pttn=(\pttn_1,\ldots,\pttn_\parts)$ and, without loss of generality, let~$|\pttn_1| = 3$. Suppose that there is at least one part consisting only of standard-agents. Such a partition is clearly not MMS, as the MMS-share of every standard-agent is one, but no two standard-agents are friends. Therefore, in every MMS partition, each part contains at least one guard-agent. By the Pigeonhole principle, we have that each part contains exactly one guard-agent. Thus, at least one guard-agent is in a part of size two with exactly one standard-agent. However, the utility of such guard-agent is~$1$, which is strictly smaller than its MMS-share, which is~$2$. We showed that for our instance, no MMS~$\parts$-partition exists.
\end{proof}

On a positive note, there is a simple structural condition that guarantees the existence of MMS partitions. Informally, if the agents do not have too many friends compared to the number of parts, an MMS partition always exists and can be found efficiently.

\begin{observation}\label{thm:MMS:exists:if:parts:greater:maxDeg}
	If~$\parts > \Delta$, an MMS partition is guaranteed to exist and can be found in linear time.
\end{observation}
\begin{proof}
	Let~$a\in\agents$ be an agent with~$\deg(a) = \Delta$ and~$\pttn$ be a~$\parts$-partition maximizing~$a$'s MMS-share. Since~$\parts > \Delta$, we obtain by the Pigeonhole principle that at least one part of~$\pttn$ contains no friend of the agent~$a$. Therefore,~$a$'s MMS-share in such instances is~$0$. Consequently, every~$\parts$-partition of~$\agents$ is MMS.
\end{proof}

\subsection{Verifying Fairness}

Now, we discuss the complexity of verifying whether a given partition is fair with respect to each notion of interest. It is easy to see that, given a partition, it can be verified in polynomial time whether this partition is EF, EFX$_0$, EFX, or EF1.  We just enumerate all pairs of agents and check whether the corresponding fairness criterion is violated or not. Similarly, it is easy to determine the~$\operatorname{PROP-share}$ for each agent and check whether their utility in the given partition exceeds this value. For MMS though, as we show in the following result, the situation is not as positive as for the other fairness notions we study.

\begin{theorem}\label{thm:MMS:NPh:verify}
	It is \NPc to decide whether the MMS-share of an agent~$a\in\agents$ is at least a given~$\sigma\in\N$, even if~$\parts=2$ and the utilities are symmetric.
\end{theorem}
\begin{proof}
	We show the result by reduction from the \probName{Equitable Partition} problem. Here, we are given a sequence~$S=(s_1,\ldots,s_{2n})$ of integers, and our goal is to decide whether there exists~$I\subseteq[2n]$ of size~$n$ such that~$\sum_{i\in I} s_i = \sum_{i\in [2n]\setminus I} s_i$. This problem is known to be \NPc~\cite{GareyJ1979}.
	
	We construct an equivalent instance~$\mathcal{J}$ of our problem as follows. The friendship graph is a star with~${2n}$ leaves~$a_1,\ldots,a_{2n}$ and a center~$c$. The leaves are in one-to-one correspondence with elements of~$S$; to ensure this, we set~$\util_c(a_i) = \util_{a_i}(c) = s_i$ for every~$i\in[2n]$. Observe that the utilities are symmetric. To finalize the construction, we set~$a=c$ and~$\sigma = \sum_{i\in[2n]} s_i / 2$.
	
	For correctness, assume that~$S$ is a \emph{yes}-instance and~$I$ is a solution partition. We create the following partition~$\pttn = (\pttn_1,\pttn_2)$, where~$\pttn_1 = \{a_i \mid i \in I\}$ and~$\pttn_2 = \{a_i \mid i \not\in I\}$. The partition is clearly balanced. Since~$I$ is a solution for~$S$, it holds that~$\sum_{i\in I} s_i = \sum_{i\in [2n]\setminus I} s_i = \sum_{i\in[2n]} s_i / 2$ and therefore also~$\util_a(\pttn_1) = \sum_{i\in I} \util_a(a_i) = \sum_{i\in I} s_i = \sum_{i\in[2n]} s_i / 2 = \sigma$ and similarly for~$\pttn_2$. Therefore,~$\mathcal{J}$ is also a \emph{yes}-instance.
	
	In the opposite direction, let~$\mathcal{J}$ be a \emph{yes}-instance, that is,~$\operatorname{MMS-share}(c) \geq \sigma = \sum_{i\in[2n]} s_i / 2$. Additionally, let~$\pttn=(\pttn_1,\pttn_2)$ be a partition such that~$\min_{i\in [\parts]} \util_a(\pttn_i) \geq \sigma$, and without loss of generality, let~$\util_a(\pttn_1) \leq \util_a(\pttn_2)$. Clearly, it holds that~$\util_a(\pttn_1) = \sigma$, as otherwise, the utility of~$a$ in~$\pttn_2$ would be smaller than in~$\pttn_1$. Consequently, we have that~$\util_a(\pttn_1) = \util_a(\pttn_2) = \sigma$. Now, we create a solution~$I$ for~$S$ by setting~$I = \{i \mid a_i \in \pttn_1 \}$. Since~$\pttn$ is balanced,~$|I|$ is of size~$n$. From the utilities, we have that~$\sum_{i\in I} s_i = \sum_{i\in I} \util_a(a_i) = \sum_{a_i \in \pttn_1} \util_a(a_i) = \sigma = \sum_{i\in[2n]} s_i / 2$. Hence,~$I$ is indeed a solution, and~$S$ is also a \emph{yes}-instance.%
\end{proof}

On the other hand, if we restrict ourselves to binary utilities, MMS-shares of each agent can be computed efficiently.

\begin{observation}
	If the utilities are binary, the MMS-share of any agent~$a\in A$ can be computed in linear time.
\end{observation}
\begin{proof}
	First, we fix an ordering~$(a_1,\ldots,a_{n-1})$ of agents in~$\agents\setminus\{a\}$, where all~$a_1,\ldots,a_{\deg(a)}$ agents are friends of~$a$. Next, we create a partition~$\pttn = (\pttn_1,\ldots,\pttn_\parts)$ by assigning the agents to parts in a round-robin way with respect to the fixed ordering of agents. Clearly, the partition is balanced and, by the Pigeonhole principle, each part contains at least~$\lfloor \deg(a)/\parts \rfloor$ friends of~$a$, that is,~$\min_{i\in [\parts]} \util_a(\pttn_i) = \lfloor \deg(a)/\parts \rfloor$. For the sake of contradiction, let there exist a partition~$\pttn^*$ with~$\min_{i\in[\parts]} \util_a(\pttn^*_i) \geq \lfloor \deg(a)/\parts \rfloor + 1$. Then, as the utilities are binary, the number of friends of~$a$ is~$\parts \cdot (\lfloor \deg(a)/\parts \rfloor + 1) = \deg(a) + \parts$, which is not possible. Therefore,~$\pttn$ achieves~$\operatorname{MMS-share}$ of~$a$ and can be clearly constructed in linear time. 
\end{proof}

\subsection{Relationships Between Studied Fairness Notions}

The definitions of envy-based fairness concepts clearly indicate the relationships between them: EF~$\Rightarrow$ EFX$_0$, EFX$_0$~$\Rightarrow$ EFX, and EFX~$\Rightarrow$ EF1, and that none of the implications is actually an equivalence.
On the other hand, though, the relations between our two share-based fairness concepts and their connections with the envy-based notions are not straightforward. The rest of the section is devoted to clarifying these relationships, which are concisely depicted in \Cref{fig:unrestricted-results-overview}.

First, we show that the existence of an EF partition does not imply the existence of a PROP partition.

\begin{proposition}
	For every~$\parts \geq 2$, there is an instance that admits an EF~$\parts$-partition and no PROP~$\parts$-partition.
\end{proposition}
\begin{proof}
	Let~$\parts \geq 2$ be a fixed number of parts and let the friendship graph be a complete graph with~$2\parts$ agents~$a_1,\ldots,a_{2\parts}$. Since there are~$2\parts$ agents, the size of each part in every balanced partition is exactly~$2$. Also, as the friendship graph is a complete graph, the PROP-share of every agent is the same:~$\frac{2\parts-1}{\parts} = 2 - \frac{1}{\parts} \in (1,2)$. Hence, proportionality requires that every agent is in part with at least two friends, which is not possible due to the balancedness requirement. Therefore, such an instance admits no PROP~$\parts$-partition. On the other hand, let~$\pttn$ be any balanced~$\parts$-partition. Clearly, each agent is in a part with precisely one of its friends, and every other part contains exactly two of its friends. However, no agent can improve by replacing an agent in any other part, as the number of friends in the target part decreases to one by the replacement. Therefore,~$\pttn$ is~EF.
\end{proof}

Next, we show that, while for~$\parts=2$ every PROP partition (if it exists) is also an EF, this no longer holds when~$\parts \geq 3$.

\begin{proposition}\label{thm:EF:PROP:equal:ifTwoParts}
	If~$\parts=2$, every PROP partition is also EF. For every~$\parts \geq 3$, there is a PROP partition which is not EF, even if the utilities are binary.
\end{proposition}
\begin{proof}
	Let~$\pttn$ be a PROP~$2$-partition and let~$a\in\agents$ be an arbitrary agent. By definition, it holds that~$\util_a(\pttn) \geq \util_a(\Fr(a))/2$. Let~$b\in\agents\setminus\pttn(a)$ be an agent of the other part. Then~$\util_a(\pttn(b)\setminus\{b\})) = \util_a(\Fr(a)) - \util_a(\pttn(a)) \leq \util_a(\Fr(a))/2$. Thus,~$\pttn$ is also an EF partition. 
	
	For the second part of the proposition, assume an instance as of \Cref{fig:EF:PROP:equalIfTwoParts} and let~$\parts = 3$. The utilities are binary.

    \begin{figure}[bt!]
        \centering
        \begin{tikzpicture}[every node/.style={draw,circle,inner sep=3pt}]
            \node (a0) at ( 0, 0) {$a_0$};
            \node (a1) at ( 2, 0) {$a_1$};
            \node (a2) at (-2,-2) {$a_2$};
            \node (a3) at ( 0,-2) {$a_3$};
            \node (a4) at ( 2,-2) {$a_4$};
            \node (a5) at ( 4,-2) {$a_5$};
            \node (a6) at ( 6,-2) {$a_6$};

            \begin{pgfonlayer}{background}
                \fill[blue!20] \convexpath{a0,a2}{15pt};
                \fill[red!20] \convexpath{a1,a4,a3}{15pt};
                \fill[green!20] \convexpath{a5,a6}{15pt};
            \end{pgfonlayer}
   
            \draw (a0) edge (a2) edge (a3) edge (a4);
            \draw (a1) edge (a3) edge (a4);
        \end{tikzpicture}
        \caption{An instance and a PROP partition which is not EF.}
        \label{fig:EF:PROP:equalIfTwoParts}
    \end{figure}
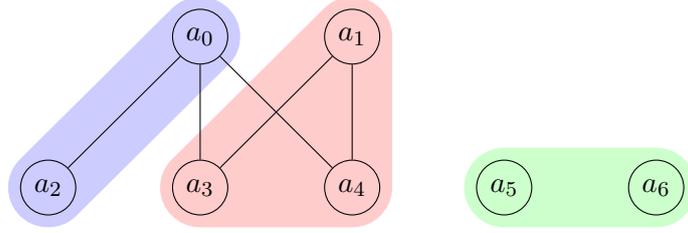
        
	The partition~$\pttn = (\{a_0,a_2\},\{a_1,a_3,a_4\},\{a_5,a_6\})$ is clearly PROP: the utility of agents~$a_2,\ldots,a_6$ is trivially larger than their PROP-share, the agent~$a_1$ is in part with all of its friends, and PROP-share of agent~$a_0$ is, as well as its utility, exactly~$1$. However, if the agent~$a_0$ replaces agent~$a_1$ in part~$\pttn_2$, its utility increases to~$2$. Therefore, the partition~$\pttn$ is not EF. To generalize the construction for any~$\parts > 3$, we just add~$\parts-3$ copies of the coalition~$\pttn_2$: paths on three agents such that the endpoints of each path are neighbors of~$a_0$. This finishes the proof.
\end{proof}

Then, we study the connection between PROP and MMS partitions.

\begin{proposition}
	For every~$\parts \geq 2$, every PROP~$\parts$-partition is MMS. Also, for every~$\parts\geq 2$, there are~$\parts$-partitions that are MMS and not PROP, even if the utilities are binary. 
\end{proposition}
\begin{proof}
	By the definition of PROP and MMS, it clearly holds that~$\operatorname{MMS-share}(a) \leq \operatorname{PROP-share}(a)$, for every agent~$a$. Therefore, if a partition~$\pttn$ is PROP, then~$\util_a(\pttn) \geq \operatorname{PROP-share}(a) \geq \operatorname{MMS-share(a)}$, which shows that~$\pttn$ is also MMS. 
 
    In the opposite direction, let~$\parts\geq 2$ be fixed. The instance contains~$\parts+1$ agents~$a_1,\ldots,a_\parts,c$, and the friendship graph is a star with~$c$ being its center. The utilities are binary. The PROP-share of every leaf~$a_i$,~$i\in[\parts]$, is exactly~$\frac{1}{\parts}$. Consequently, they all need to be in the same part as agent~$c$. However, this is not possible due to the balancedness requirement. Therefore, no PROP allocation exists in such an instance. The existence of MMS partition follows by the fact that~$\Delta = \parts$ combined with \Cref{thm:MMS:exists:if:parts:greater:maxDeg}.
\end{proof}

Next, we show the analog of \Cref{thm:EF:PROP:equal:ifTwoParts}, this time though for MMS and EF1.

\begin{proposition}
	If~$\parts=2$, every MMS partition is also EF1. For every~$\parts \geq 3$, there is an instance that admits an MMS~$\parts$-partition which is not EF1, even if the utilities are binary.
\end{proposition}
\begin{proof}
	For the first part of the proposition, let~$\pttn$ be an MMS~$2$-partition. Suppose that~$\pttn$ is not EF1, that is, there is a pair of agents~$a,b$ such that~$\util_a(\pttn(a)) < \util_a(\pttn(b)\setminus\{b\}) - \util_a(c)$ for every agent~$c\in \pttn(b)\setminus\{b\}$. 
    This means two things. 
    First, {\em every} agent~$c\in \pttn(b)\setminus\{b\}$ is a friend of~$a$. 
    Second, there exists an agent~$d \in \pttn(a)$, that~$a$ values strictly less than any agent~$c \in \pttn(b)$.
    If any of the above was not true, then the partition {\em would} be EF1. But then agent~$a$ could swap agent~$d\in \pttn(a)$ with any agent~$c \in \pttn(b)$ and get that~$\util_a((\pttn(a)\setminus\{d\})\cup c) > \util_a(\pttn(a)$ and~$\util_a((\pttn(b)\setminus\{b,c\})\cup d) > \util_a(\pttn((\pttn(b)\setminus\{b\}))$. This contradicts the initial assumption that the starting 2-partition was an MMS-partition.
	
	Let~$\parts \geq 3$ be a fixed number of parts. Our instance consists of~$2k+1$ agents~$c,\ell_1,\ldots,\ell_{\parts-1},g_1,\ldots,g_{\parts+1}$ and the set of edges of the friendship graph~$G$ contains the edge~$\{c,\ell_i\}$ for every~$i\in[\parts-1]$. That is, the friendship graph is a disjoint union of a star with center~$c$ and~$\parts$ leaves and an edgeless graph with~$\parts-1$ vertices. Observe that the~$\operatorname{MMS-share}$ of every agent is zero. We partition the agents such that~$\pttn_1 = \{c,g_{\parts+1}\}$,~$\pttn_2 = \{\ell_1,\ell_2,\ell_3\}$, and the rest of the agents are partitioned arbitrarily into parts of size~$2$. This partition is clearly MMS by \Cref{thm:MMS:exists:if:parts:greater:maxDeg}. However,~$\pttn$ is not EF1: the current utility of agent~$c$ is~$0$ and the utility of~$c$ after replacing~$\ell_1$ in~$\pttn_2$ and removing any remaining agent of~$\pttn_1\setminus\{\ell_1\}$ is~$3-2 = 1$. That is, the envy towards~$\ell_1$ cannot be eliminated by the removal of a single friend.
\end{proof}

Our last observation, which concludes the section, states that EF1 does not imply MMS.

\begin{observation}
    For every~$\parts\geq 2$, there exists an instance that admits an EF1 partition and no MMS partition, even if the utilities are binary.
\end{observation}
\begin{proof}
	Let~$\parts\geq 2$ be a fixed number of parts. Recall the construction from \Cref{thm:MMS:nonexistence}. It was shown that such an instance does not admit any MMS~$\parts$-partition. Consider now the following~$\parts$-partition~$\pttn=(\pttn_1,\ldots,\pttn_\parts)$: for every~$j\in[2,\parts]$, we set~$\pttn_j = \{a_j,i_j\}$, and we set~$\pttn_1 = \{a_1,i_1,i_{\parts+1}\}$. Clearly,~$\pttn$ is a balanced~$\parts$-partition. As every part contains exactly one guard-agent and every standard-agent belongs to a part with exactly one guard-agent, there is no envy from standard-agents towards any other agent. Also, there is no envy from agent~$a_1$ towards any other agent, as its current utility is~$2$ and by moving to any other coalition, its utility can only decrease. On the other hand, every guard-agent~$a$,~$a\not\in\pttn_1$, has envy towards any agent in the part~$\pttn_1$. Let~$b,c\in\pttn_1$ be a pair of distinct agents of~$\pttn_1$. Then, we have~$\util_a(\pttn(b)\setminus\{b\}) = 2 = \util_a(\pttn(a)) + 1 = \util_a(\pttn(a)) + \util_a(c)$. Thus,~$\pttn$ is EF1.
\end{proof}

\section{Algorithms and Complexity}\label{sec:complexity}

\begin{figure}[bt!]
	\centering
    \begin{tikzpicture}[every node/.style={minimum width=1.5cm}]
		\node[draw] (EF) at (1.5,0) {EF};
		\node[draw] (EFX0) at (1.5,-1) {EFX$_0$};
		\node[draw] (EFX) at (1.5,-2) {EFX};
		\node[draw] (EF1) at (1.5,-3) {EF1};
		\node[draw] (PROP) at (-2,0) {PROP};
		\node[draw] (MMS) at (-2,-3) {MMS};
		
		\draw (EF) edge[->] (EFX0);
		\draw (EFX0) edge[->] (EFX);
		\draw[->] (PROP) edge (MMS);
		\draw[->] (EFX) edge (EF1);
		\draw[->,dashed] (MMS) -- (EF1) node[midway,above] {if~$k=2$};
		\draw[->,dashed] (PROP.east) -- (EF.west) node[midway,above,sloped] {if~$k=2$};

		\node[red,below of=MMS,text width=2.75cm,align=center,node distance=1cm] {\small no poly-time algo. unless \P{}=\NP};
		\begin{scope}[on background layer]
			\fill[red!40,opacity=0.5,rounded corners=5]
			(-1,-2.6) rectangle (-3,-3.4);
		\end{scope}

		\node[red,below of=EF1,text width=3cm,align=center,node distance=1cm] {\small \NP-hard even if~${\vc = \parts = 2}$}; 
		\begin{pgfonlayer}{background}
			\fill[red!40,opacity=0.5,rounded corners=5]
			(0.5,-2.6) rectangle (2.5,-3.4);
		\end{pgfonlayer}
	
		\node[red,right of=EFX,text width=2cm,align=center,node distance=2.75cm,yshift=0cm] (l4) {\small \NP-hard even if~$G$ bipartite and ${\vc = \parts = 2}$}; 
		\begin{pgfonlayer}{background}
			\fill[red!40,opacity=0.5,rounded corners=10]
				(-3.25,0.75) rectangle (2.75,-2.5);
		\end{pgfonlayer}
		\draw[red!40,opacity=0.5,ultra thick] (l4) -- (2.75,-2);
		
		\node[red,left of=PROP,node distance=2.75cm,text width=1.5cm,align=center] (l5) {\small \NP-hard\\ on paths};
		\begin{pgfonlayer}{background}
			\fill[red!40,opacity=0.9,rounded corners=10]
				(-3,0.5) rectangle (2.5,-1.45);
		\end{pgfonlayer}
		
		\draw[ultra thick, red!40] (l5) -- (PROP);
	\end{tikzpicture}
	\caption{A basic overview of our algorithmic and complexity results for unrestricted friendship graphs. An arrow from notion~$A$ to~$B$ denotes that if a partition is fair with respect to~$A$, then it is also fair with respect to~$B$. Here,~$\parts$ is the number of parts,~$n$ is the number of agents, and~$\vc$ denotes the vertex cover number of the friendship graph~$G$.}
	\label{fig:unrestricted-results-overview}
\end{figure}
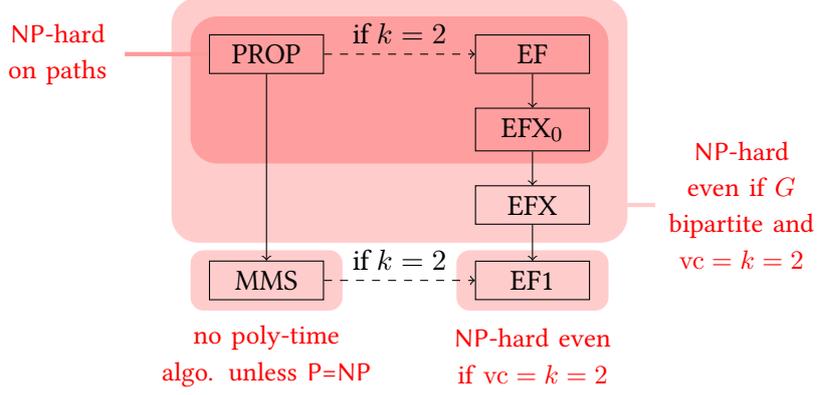

In this section, we provide the complexity landscape for the fair partitioning of friends with respect to all of the above-defined notions of fairness. 
A basic overview of our results is provided in \Cref{fig:unrestricted-results-overview}.

We start with the case of the most general fairness notions, which are envy-freeness and proportionality. Specifically, we show that in this setting, even if the friendship graph is as simple as a path, the problem is intractable. 

\begin{theorem}\label{thm:EF:NPh}\label{thm:PROP:NPh}
	It is \NPc and \Wh when parameterized by the number of parts~$\parts$ to decide whether an EF, EFX$_0$, or PROP allocation exists, even if~$G$ is a path and the utilities are objective. For EF and EFX$_0$, the hardness holds even if the utilities are encoded in unary.
\end{theorem}
\begin{proof}
	We prove the theorem by a reduction from the \probName{Unary Bin Packing} problem. In this problem, we are given a multi-set~$S=(s_1,\ldots,s_N)$ of positive integers, a number of bins~$B$, and a capacity~$c$ of every bin. Without loss of generality, we can assume that~$\sum_{i\in[N]} s_i = c\cdot B$. The goal is to find an allocation~$\alpha\colon S\to [B]$ such that for every~$i\in[B]$ it holds that~$\sum_{j\colon \alpha(s_j) = i} s_j = c$. The \probName{Unary Bin Packing} problem is known to be \NPc and \Wh when parameterized by the number of bins~$B$, even if all the integers are encoded in unary,~$c \geq 4$,~$B\geq3$, and every~$s_i\in S$ is at least two~\cite{JansenKMS2013}.
	
	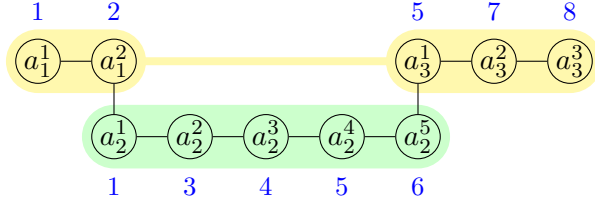
\begin{figure}
		\centering
		\begin{tikzpicture}[every node/.style={draw,circle,inner sep=0.25pt,},label distance=2mm]
			\node[label=90:\textcolor{blue}{\small$1$}] (a11) at (-3,0) {$a_1^1$};
			\node[label=90:\textcolor{blue}{\small$2$}] (a12) at (-2,0) {$a_1^2$};

            \begin{pgfonlayer}{background}
    			\fill[yellow!40] \convexpath{a11,a12}{12pt};
    		\end{pgfonlayer}
			
			\node[label=270:\textcolor{blue}{\small$1$}] (a21) at (-2,-1) {$a_2^1$};
			\node[label=270:\textcolor{blue}{\small$3$}] (a22) at (-1,-1) {$a_2^2$};
			\node[label=270:\textcolor{blue}{\small$4$}] (a23) at (0,-1) {$a_2^3$};
			\node[label=270:\textcolor{blue}{\small$5$}] (a24) at (1,-1) {$a_2^4$};
			\node[label=270:\textcolor{blue}{\small$6$}] (a25) at (2,-1) {$a_2^5$};

            \begin{pgfonlayer}{background}
    			\fill[green!40,opacity=0.5] \convexpath{a21,a25}{12pt};
    		\end{pgfonlayer}
   
			\node[label=90:\textcolor{blue}{\small$5$}] (a31) at (2,0) {$a_3^1$};
			\node[label=90:\textcolor{blue}{\small$7$}] (a32) at (3,0) {$a_3^2$};
			\node[label=90:\textcolor{blue}{\small$8$}] (a33) at (4,0) {$a_3^3$};

            \begin{pgfonlayer}{background}
    			\fill[yellow!40] \convexpath{a31,a33}{12pt};
                \fill[yellow!40] \convexpath{a12.east,a31.west}{1.5pt};
    		\end{pgfonlayer}
			
			\draw (a11) -- (a12) -- (a21) -- (a22) -- (a23) -- (a24) -- (a25) -- (a31) -- (a32) -- (a33);
		\end{tikzpicture}
		\caption{An illustration of the construction used to prove \Cref{thm:EF:NPh}. In this example, we assume an instance of \probName{Unary Bin Packing} with~$S=(2,5,3)$,~$B=2$, and~$c=5$. Next to each agent~$a_i^j$, we depict (in blue) the value of~$\util_\forall(a_i^j)$. With differently colored backgrounds, we highlight one possible balanced~$2$-partition for this instance.}
		\label{fig:EF:NPh}
	\end{figure}

	We start with the case of EF and EFX$_0$ partitions, and later, we show how to tweak the construction to prove the result also for PROP. Given an instance~$\mathcal{I}=(S,B,c)$, we create an equivalent instance~$\mathcal{J}$ of our problem as follows. For every element~$s_i\in S$, we introduce a set~$\agents_i$ of~$s_i$ agents~$a_i^1,\ldots,a_i^{s_i}$ and make each two consecutive agent friends; that is, we add an edge~$\{a_i^j,a_i^{j+1}\}$ for every~$j\in[s_i - 1]$. It is easy to see that these agents induce a path. To finalize the construction of~$G$, we add an edge~$\{ a_i^{s_i}, a_{i+1}^1 \}$ for every~$i\in[N-1]$. Observe that~$G$ is a single path with exactly~$c\cdot B$ agents. Next, we define the utilities. First, we set~$\util_\forall(a_1^i) = i$ for every~$i\in[s_1]$. For every~$i\in[2,N]$, we define the utilities recursively as follows
	\[
	\util_\forall(a_i^j) = \begin{cases}
		\util_\forall(a_{i-1}^{s_{i-1}-1}) & \text{if } j = 1 \text{,} \\
		\util_\forall(a_{i}^{j-1}) + 2 & \text{if } j = 2 \text{, and} \\
		\util_\forall(a_{i}^{j-1}) + 1 & \text{otherwise.}
	\end{cases}
	\]
    \Cref{fig:EF:NPh} illustrates our construction. Observe that in our construction, the utilities are of size~$\Oh{n}$; that is, we do not require utilities of exponential size. The basic idea behind the construction is that once an agent~$a_i^{s_i}$,~$i\in[N]$, is assigned to some part~$\pttn_j$ in a partition~$\pttn$, all~$a_i^j$,~$j\in[s_i - 1]$ are necessarily assigned to the same part, as otherwise,~$\pttn$ is not EF. To complete the construction, we set~$\parts = B$.

	For correctness, assume that~$\mathcal{I}$ is a \emph{yes}-instance and~$\alpha$ is a solution allocation. We create a partition~$\pttn=(\pttn_1,\ldots,\pttn_\parts)$ such that for every~$i\in[\parts]$, we set~$\pttn_i = \bigcup_{j\colon \alpha(s_j) = i} \{ a_j^1,\ldots,a_j^{s_j} \}$. It is clear that the partition is balanced, as~$\alpha$ is a solution allocation. To see that~$\pttn$ is EF, assume an agent~$a_i^j$ for some~$i\in[N]$ and~$j\in[s_i]$. By the definition of~$\pttn$,~$a_i^j$ is in the same part as all other agents of~$\agents_i$. If~$j \in [2,s_i - 1]$, then such an agent is clearly not envious, as it is in the same part as all of his friends. If~$j=1$, then~$a_i^{j+1}$ is in the same part as~$a_i^j$ and even if~$a_{i-1}^{s_{i-1}}$ is in a different part,~$a_i^j$ is not envious since~$\util_\forall(a_{i}^{j+1}) > \util_\forall(a_{i-1}^{s_{i-1}})$. Finally, if~$j=s_i$, then~$\util_\forall(a_{i}^{j-1}) = \util_\forall(a_{i+1}^1)$ and~$a_i^j$ is in the same part with at least~$a_i^{j-1}$ and therefore, is also not envious towards any agent.
 
	In the opposite direction, let~$\mathcal{J}$ be a \emph{yes}-instance and let~$\pttn=(\pttn_1,\ldots,\pttn_\parts)$ be an EFX$_0$ partition. First, we show that for every~$A_i$,~$i\in[N]$, it holds that~$\agents_i \subseteq \pttn_j$ for some~$j\in[\parts]$. In other words, all agents corresponding to a single item~$s_i\in S$ are in the same part in every EFX$_0$ solution. For the sake of contradiction, assume that it is not the case, and we have~$\agents_i$ such that~$\pttn(a) \not= \pttn(a')$ for some~$a,a'\in\agents_i$. Let~$j \in [2,s_i]$ be the smallest~$j$ such that~$\pttn(a_i^1) \not= \pttn(a_i^j)$. Since we selected~$j$ to be the smallest possible, it also holds that~$\pttn(a_i^{j-1}) \not= \pttn(a_i^{j})$. Focus now on the utility of the agent~$a_i^{j-1}$. By the definition of the utilities, we have~$\util_\forall(a) < \util_\forall(a_i^j)$ for every~$a\in \Fr(a_i^{j-1})$. However,~$a_i^{j-1}$ is not in the same part as~$a_i^j$, so it is envious towards any agent~$b$ in~$\pttn(a_i^j)\setminus\{a_i^j\}$ and, by the assumption that~$c \geq 4$, there is at least one more non-friend of~$a_i^j$ and thus, the envy cannot be eliminated by the removal of any agent. This contradicts that~$\pttn$ is EFX$_0$ and necessarily, for every~$i\in[N]$, it holds that~$A_i\subseteq \pttn_j$ for some~$j\in[\parts]$ in every EFX$_0$ partition. 

    We define an allocation~$\alpha$ so that for every~$s_i$, we set~$\alpha(s_i) = j$, where~$j$ is the number of part such that~$A_i\subseteq \pttn_j$. Clearly,~$\alpha$ is well-defined for all elements of~$S$. For the sake of contradiction, let there exist a bin~$i\in[k]$ such that~$\sum_{s_j\colon \alpha(s_j) = i} s_j \not= c$. Recall that for every~$\ell\in[N]$, it holds that all agents of~$\agents_\ell$ are in the same part. Therefore, the size of~$\pttn_i$ is not~$c$, which contradicts that~$\pttn$ is balanced. This is not possible as~$\pttn$ is a solution partition. Hence,~$\alpha$ is indeed a solution for~$\mathcal{I}$.
	
	It is easy to see that the construction can be performed in polynomial time. Moreover, we used~$\parts = B$, and therefore, our reduction is also a parameterized reduction. This finishes the \NPhness and \Whness with respect to~$\parts$ of deciding whether an EF or EFX$_0$ partition exists.
	
	For PROP, we tweak the utilities in the construction as follows. We set~$\util_\forall{a_1^i} = \parts^i$ for every~$i\in [s_1]$ and for every~$i\in[2,N]$, we set 
	\[
	\util_\forall(a_i^j) = \begin{cases}
		\util_\forall(a_{i-1}^{s_{i-1}-1}) & \text{if } j = 1 \text{,} \\
		\util_\forall(a_{i}^{j-1}) \cdot k^2 & \text{if } j = 2 \text{, and} \\
		\util_\forall(a_{i}^{j-1}) \cdot k & \text{otherwise.}
	\end{cases}
	\]
	Now,~$\operatorname{PROP-share}$ of every agent~$a$ is exactly~$k^\ell + k^{\ell-1}$ and the utilities of its friends all either~$k^\ell$ or~$k^{\ell+1}$. Therefore, each agent is necessarily in the same part as its friend it values more, which gives us the same property we exploited in the case of EF and EFX$_0$. Otherwise, the hardness proof is the same.
\end{proof}

If we move away from paths, then our problem becomes intractable even for~$\parts=2$, and the friendship graph is bipartite with vertex cover number 2.

\begin{theorem}\label{thm:EFX:NPh}
    It is \NPc to decide whether a PROP, EF, EFX$_0$, or EFX partition exists, even if~$\parts = 2$,~$G$ is a bipartite graph,~$\vc(G)=2$, and the utilities are objective. 
\end{theorem}
\begin{proof}
	We again reduce from \textsc{Equitable Partition}. Recall that in this problem, we are given a sequence~$S=(s_1,\ldots,s_{2N})$ of integers, and our goal is to decide whether there exists~$I\subseteq[2N]$ of size~$N$ such that~$\sum_{i\in I} s_i = \sum_{i\in [2N]\setminus I} s_i = B$. 
    This problem is known to be \NPc even if each~$s_i\in S$ is at least~$N^2$ and~$\max S - \min S < \frac{\min S}{N^2}$ \cite{DeligkasEKS2024}. Without loss of generality, we assume that~$N \geq 10$, as otherwise, we can solve the instance trivially in polynomial time.
    
	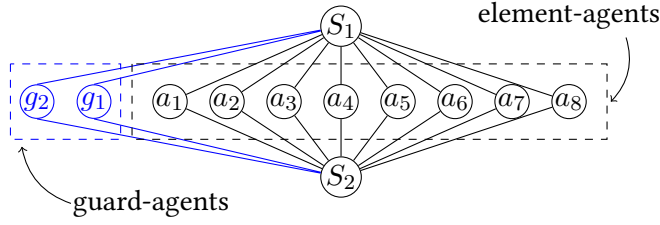
\begin{figure}
		\centering
		\begin{tikzpicture}
			\node[draw,circle,inner sep=0.2pt] (S1) at (0,0) {$S_1$};
			\node[draw,circle,inner sep=0.2pt] (S2) at (0,-2) {$S_2$};
			
			\foreach[count=\i] \x in {-2.25,-1.5,...,2.25,3}{
				\node[draw,circle,inner sep=0.2pt] (a\i) at (\x,-1) {$a_\i$};
				\draw (S1) -- (a\i) -- (S2);
			}
			\node[] (ellabel) at (3,0.15) {element-agents};
			\draw (3.75,-0.15) edge[bend left,->] (3.55,-1);
			
			\node[draw,circle,inner sep=0.2pt,blue] (g1) at (-3.25,-1) {$g_1$};
			\node[draw,circle,inner sep=0.2pt,blue] (g2) at (-4,-1) {$g_2$};
			\node[inner sep=0.5pt] (glabel) at (-2.5,-2.35) {guard-agents};
			\draw (glabel.west) edge[bend left,->] (-4.2,-1.6);
			\draw[blue] (S1) -- (g1.north) (g1.south) -- (S2) -- (g2.south) (g2.north) -- (S1);
			
			\draw[dashed,blue] (-4.35,-0.5) rectangle (-2.9,-1.5);
			\draw[dashed] (-2.75,-0.5) rectangle (3.5,-1.5);
			
		\end{tikzpicture}
		\caption{An illustration of the construction used in \Cref{thm:EFX:NPh}. For guard-agents, which are depicted in blue, we have~$\util_\forall(g_1) = \util_\forall(g_2) = 1$, and~$\util_\forall(S_1) = \util_\forall(S_2) = B$, where~$B$ is half of the sum of all integers given in an original instance of the \probName{Equitable Partition} problem. The element-agents are in one-to-one correspondence with integers of the \probName{Equitable Partition} problem, which is captured in its objective utility.}
		\label{fig:EFX:NPh}
	\end{figure}
	
	Given an instance~$S$ of \probName{Equitable Partition}, we create an equivalent instance~$\mathcal{J}$ of our problem as follows. First, we create an \emph{element-agent}~$a_i$ for every~$s_i\in S$ and we set~$\util_\forall(a_i) = s_i$. Next, we add two \emph{guard-agents}~$g_1$ and~$g_2$ such that~$\util_\forall(g_1) = \util_\forall(g_2) = 1$. In addition, we introduce two \emph{set-agents}~$S_1$ and~$S_2$ and define~$\util_\forall(S_1) = \util_\forall(S_2) = B$. The friendship graph is the complete bipartite graph~$K_{2,2N+2}$, where the smaller part consists of the set-agents; see \Cref{fig:EFX:NPh} for an illustration of the friendship graph. To finalize the construction, we set~$\parts = 2$.

    For correctness, let~$S$ be a \emph{yes}-instance and let~$I$ be a solution. We claim that partition~$\pttn=(\pttn_1,\pttn_2)$, where~$\pttn_1 = \{S_1,g_1\} \cup \{ a_i \mid i \in I \}$ and~$\pttn_2 = \{S_2, g_2\} \cup \{ a_i \mid i \in [2N]\setminus I \}$, is PROP. First, the utility of element-agents and the guard-agents with respect to~$\pttn$ is~$B$, and their PROP-share is also~$B$. Therefore, the partition is PROP for them. The utility of~$S_1$ is~$1 + \sum_{i\in I} \util_\forall(a_i) = 1 + \sum_{i \in I} s_i = 1 + \sum_{i\in [2N]\setminus I} s_i/2  = 1 + B$ and so is its PROP-share. By similar arguments, we obtain that~$\pttn$ is PROP also for~$S_2$. By \Cref{thm:EF:PROP:equal:ifTwoParts}, we have that~$\pttn$ is also EF, and therefore, it is also EFX$_0$ and EFX.
	
	In the opposite direction, let~$\mathcal{J}$ be a \emph{yes}-instance for our problem. Before we focus on a specific fairness notion, we show several properties that hold for all of them. First, we show that the set-agents are necessarily in different parts.
	
	\begin{claim}\label{Claim1}
		In every fair partition~$\pttn$, we have~${\pttn(S_1) \not= \pttn(S_2)}$.
	\end{claim}
	\begin{claimproof}
        For the sake of contradiction, assume that~$S_1$ and~$S_2$ are in the same part, without loss of generality,~$\pttn_1$. 
        Then, due to balancedness, there necessarily exists a pair of distinct agents~$a\in \agents\setminus\{S_1,S_2\}$ such that~$a \in \pttn_2$. For these agents,~$\Fr(a) = \{S_1,S_2\}$. Therefore,~$\util_a(\pttn) = 0 < \operatorname{PROP-share}(a) = B$. Hence,~$\pttn$ is not PROP. Moreover, since~$n \geq 10$, there exists an agent~$b\in\agents\setminus\{S_1,S_2\}$ such that~$b\in\pttn_1$. If~$a$ replaces~$b$ in~$\pttn_1$, the utility is~$\util_a(\pttn_1\setminus\{b\}) = \util_a(S_1) + \util_a(S_2) = 2B$. This envy clearly cannot be eliminated by the removal of any friend. This contradicts that~$\pttn$ is EFX. Therefore,~$S_1$ and~$S_2$ are necessarily in different parts.
	\end{claimproof}
	
	So, without loss of generality, for the rest of the proof, we assume that~$S_1 \in \pttn_1$ and~$S_2 \in \pttn_2$. Now, we show a similar property also for the guard-agents.

	\begin{claim}\label{Claim2}
		In every fair partition~$\pttn$, we have~$\pttn(g_1) \not= \pttn(g_2)$.
	\end{claim}
	\begin{claimproof}
       We will prove the claim by contradiction. So, for the sake of contradiction assume that there is a fair allocation where both guard-agents~$g_1$ and~$g_2$ belong to~$\pi_1$. 
    Then, due to balancedness and due to \Cref{Claim1}~$\pi_1$ contains~$N-1$ element-agents. Thus, we have that  
    \begin{align*}
    \util_{S_1}(\pi_1) 
        & \leq (N-1) \max S + 2\\
        & < (N-1) \min S + \frac{N-1}{N^2} \cdot \min S + 2 \\
        & \leq N^3 - N^2 + N + 1\\
        & < N^3
    \end{align*}
    where the second inequality follows from our assumption for the \probName{Equitable Partition} instance, the third one from the assumption that~$\min S \geq N^2$, and the last one from the assumption that~$N \geq 10$.

    Observe now the following. Firstly,~$\pi_2$ contains~$N+1$ element agents and agent~$S_2$.
    Furthermore, if agent~$S_1$ moves to~$\pi_2$ by (a) exchanging its position with set-agent~$S_2$ and (b) excluding some element-agent~$a_i$ from~$\pi_2$, then~$\util_{S_1}(\pi_2\setminus \{S_2, a_i\}) \geq N \min S  > N^3 > \util_{S_1}(\pi_1)$, since we have assumed that~$\min S \geq N^2$. Hence, the assumed partition cannot be EFX, thus it cannot be EFX$_0$, or EF.

    For PROP now, observe that~$\operatorname{PROP-share}$ for both set-agents~$S_1$ and~$S_2$ is~$B+1$. 
        From \Cref{Claim1} both guard-agents~$g_1,g_2$ belong to the same group with~$S_1$ and since there are~$2N+4$ agents, part~$\pi_1$ must contain~$N-1$ element-agents. Consequently~$\util_{S_1}(\pi_1) = 2+\sum_{a_i\in\pi_1}\util_{S_1}(a_i)$ and
       ~$\util_{S_2}(\pi) = \sum_{a_j\in\pi_2}\util_{S_2}(a_j)$. 
        If~$\pi$ is PROP then it must be the case that ~$2+\sum_{a_i\in\pi_1}\util_{S_1}(a_i)\geq 1+B$ and~$\sum_{a_j\in\pi_2}\util_{S_2}(a_j)\geq 1+B$. 
        Thus,~$\sum_{a_i\in \pi_1}s_i = B-1$ and~$\sum_{a_i\in \pi_2}s_i = B+1$. 
        This is a contradiction, because from \cite{DeligkasEKS2024} no sum of any~$N-1$ elements of~$S$ can be close to~$B$. Thus if partition~$\pi$ is PROP then~$\pi(g_1)\neq \pi(g_2)$.
\end{claimproof}
	
	As the final auxiliary claim, we show that the utility of the set-agents is exactly the same in every fair solution partition.
	
	\begin{claim}\label{thm:setAgents:sameUtils}
		In every fair  partition~$\pttn$, we have~$\util_{S_1}(\pttn) = \util_{S_2}(\pttn) = 1 + B$.
	\end{claim}
	\begin{claimproof}
    Given the previous claims, assume without loss of generality that agents~$S_1$ and~$g_1$ are part of~$\pi_1$ and that agents~$S_2$ and~$g_2$ belong in~$\pi_2$. Let the remaining~$2N$ element-agents be divided into~$\pi_1$ and~$\pi_2$, such that each group contains~$N$ agents. As we computed,~$\operatorname{PROP-share}(S_1) = \operatorname{PROP-share}(S_2) = 1+B$. Thus if~$\pi$ is PROP then~$\util_{S_1}(\pi_1) \geq 1+B$ and~$\util_{S_2}(\pi_2)\geq 1+B$. Since~$\util_{S_1}(\pi) = 1+\sum_{a_i\in \pi_1}\util_{S_1}(a_i) = 1+\sum_{i\in I}s_i$ and~$\util_{S_2}(\pi) = 1+\sum_{a_j\in \pi_2}\util_{S_2}(a_j) = 1+ \sum_{j\in[2N]/I}s_j$, it must hold that~$\sum_{i\in I}s_i \geq B$ and~$\sum_{j\in[2N]/I} s_j \geq B$. From this fact it is clear that~$\sum_{a_i\in\pi_1}\util_{S_1}(a_i) = \sum_{a_j\in\pi_2}\util_{S_2}(a_j) = B$, which proves that~$\util_{S_1}(\pi_1) = \util_{S_2}(\pi_2) = 1 + B$
     
    Next we prove the claim for EFX and thus it must hold for the stronger notions of EFX$_0$ and EF.
    So, fix agent~$S_1$ and observe that by removing any element-agent~$a_j\in\pi_2$ we get that~$\util_{S_1}(\pi_2\setminus\{S_2\}) - \util_{S_1}(a_k) = 1+\sum_{j\in[2N]/I}s_j-s_k$. The valuation of agent~$S_1$ for part~$\pi_1$ is~$\util_{S_1}(\pi_1) = 1 + \sum_{i\in I}s_i$, thus under the EFX definition it must hold that~$\sum_{i\in I}s_i \geq \sum_{k\in[2N]/I} s_k-s_j$. As we argued earlier no sum of~$N-1$ elements can be close to~$B$, thus the inequality must trivially hold after removing any element-agent from~$\pi_2$. The case for set-agent~$S_2$ is symmetrical. By removing the guard-agent~$g_2$ from~$\pi_2$, we get~$\util_{S_1}(\pi_2\setminus\{S_2\})-\util_{S_1}(g_2) = \sum_{a_k\in\pi_2}\util_{S_1}(a_k) = 1 + \sum_{k\in[2N]/I}s_k$. Thus from the definition of EFX it must hold that~$1 + \sum_{i\in I}s_i\geq \sum_{k\in[2N]/I}s_k$. Symmetrically removing the guard-agent~$g_1$ from~$\pi_1$ and applying the EFX definition on set-agent~$S_2$ it must also hold that~$1+ \sum_{k\in[2N]/I}s_k\geq \sum_{i\in I}s_i$. Combining these facts and after calculations we get that~$|\sum_{i\in I}s_i-\sum_{k\in[2N]/I}s_k|\leq 1$. Since by \cite{DeligkasEKS2024} it cannot be the case that the two terms differ by 1, we get that~$\sum_{k\in[2N]/I}s_k = \sum_{i\in I}s_i = B$. Consequently it must be the case that if~$\pi$ is EFX then~$\util_{S_1}(\pi) = \util_{S_2}(\pi) = 1+B$.
   \end{claimproof}
	
	Again, without loss of generality, we can assume that~$g_1 \in \pttn_1$ and~$g_2 \in \pttn_2$. Now, we create a solution for~$\mathcal{I}$. We set~$I = \{ i \in [2N] \mid a_i \in \pttn_1 \}$, and we show that it is indeed a solution. For the sake of contradiction, suppose that this is not the case, that is,~$\sum_{i\in I} s_i \not= \sum_{i\in [2N]\setminus I} s_i$ and, without loss of generality~$\sum_{i\in I} s_i < \sum_{i\in [2N]\setminus I} s_i$. Then the utility of the agent~$S_1$ in~$\pttn$ is~$1 + \sum_{i\in I} \util_{S_1}(a_i) = 1 + \sum_{i\in [2N]\setminus I} s_i < 1 + B/2$, which contradicts \Cref{thm:setAgents:sameUtils}. Therefore, such a situation cannot occur,~$I$ is a solution for~$\mathcal{I}$, and~$\mathcal{I}$ is thus also a \emph{yes}-instance.
	
	The construction can be clearly done in polynomial time. It is also easy to see that the friendship graph is bipartite with~$A=\{S_1,S_2\}$ and~$B = V(G)\setminus\{S_1,S_2\}$ and hence if we remove~$A$ from~$G$, we obtain an edgeless graph. Consequently, the vertex cover number of this graph is exactly~$2$, as was promised.
\end{proof}

While \Cref{thm:MMS:NPh:verify} shows that we cannot compute in polynomial time the MMS-share of an agent, it does not exclude the existence of polynomial-time algorithms for finding an MMS allocation. In fact, for the instance constructed in \Cref{thm:MMS:NPh:verify}, where the friendship graph was a star, we could efficiently compute a solution; see \Cref{thm:vc1graphs:polytime} for a formal proof. 
However, a polynomial-time algorithm is unlikely to exist in general.

\begin{theorem}\label{thm:MMS:nopolytime}
	Unless~$\P=\NP$, there is no polynomial-time algorithm that computes an MMS allocation, if it exists, even if~$G$ is a bipartite graph and~$\vc(G) = 2$.
\end{theorem}
\begin{proof}
	Let~$S$ be an instance of the \probName{Equitable Partition} problem. For such~$S$, we create the same instance of our problem as in \Cref{fig:EFX:NPh}, just with guard-agents removed. First, observe that for every element-agent~$a_i\in V(G)$ we have~$\operatorname{MMS-share}(a_i) = B$, as each element-agent is friend with both set-agents. Consequently, in every MMS partition, each element-agent is in a part with at least one set-agent, and therefore, the set-agents are not in the same part.
	
	First, we show that for this instance an MMS allocation always exists. This is not too hard to see. Observe that a partition that puts element-agents as evenly as possible is indeed an MMS allocation; this way, the ``regret'' of a set-agent is minimized.
	
	Now, suppose that there is a polynomial-time algorithm~$\mathcal{A}$ that finds an MMS allocation.  Then, this algorithm would have to split the element-agents as evenly as possible. 
    Thus, if the~$\operatorname{MMS-share}$ of set-agent~$S_1$ is~$B$, then there exists an  equi-partition~$(\pi_1, \pi_2)$ of the element agents such that~$u_{S_1}(\pi_1) = u_{S_1}(\pi_2) = B$; this is due to the fact that the elements of \probName{Equitable Partition} sum up to~$2B$. Thus, the initial instance of \probName{Equitable Partition} instance is a {\em yes}-instance.
    Conversely, if the~$\operatorname{MMS-share}$ of set-agent~$S_1$ is strictly less than~$B$, then there is no balanced partition~$(\pi_1, \pi_2)$ of the element agents such that~$u_{S_1}(\pi_1) = u_{S_1}(\pi_2) = B$; if there was such a partition, then the~$\operatorname{MMS-share}$ of~$S_1$ would be~$B$. 
    Hence, in that case, we can conclude that the initial instance of \probName{Equitable Partition} instance is a {\em no}-instance.
    
	So, overall, if there was a polynomial-time algorithm that could compute an MMS allocation for the created instance, then it could solve the initial instance of \probName{Equitable Partition}. Our theorem follows.
\end{proof}

Next, we establish hardness for EF1 as well by slightly modifying the construction from \Cref{fig:EFX:NPh}.

\begin{theorem}\label{thm:EF1_NPh_vc2}
    It is \NPc to decide whether an EF1 partition exists, even if~$\vc(G) = 2$ and~$\parts=2$. 
\end{theorem}
\begin{proof}
    Again, we reduce from \textsc{Equitable Partition}. We employ the same construction of \Cref{thm:EFX:NPh} by adding an additional edge in the friendship graph of \Cref{fig:EFX:NPh} among the set-agents~$S_1$ and~$S_2$ with~$\util_{S_1}(S_2) = \util_{S_1}(S_2) = B$. Let S be a \emph{yes}-instance of \probName{Equitable Partition} and assume that~$I$ is a solution for~$S$. We will prove that partition~$\pi = (\pi_1, \pi_2)$ where~$\pi_1 = \{S_1,g_1\}\cup \{a_i|i\in I\}$ and~$\pi_2 = \{S_2,g_2\}\cup \{a_j|j\in [2N]/I\}$ is EF1. Towards that, we will prove that the EF1 condition holds for any pair of agents.

    \begin{itemize}
    \item For any element-agent~$a_i\in \pi_1$, it holds that~$\util_{a_i}(\pi_1) = \util_{a_i}(S_1) = B$, while~$\util_{a_i}(\pi_2/\{S_2\}) = 0$,~$\util_{a_i}(\pi_2/\{g_2\}) = B$ and for every element-agent~$a_j\in\pi_2$~$\util_{a_i}(\pi_2/\{a_j\}) = B$. We see that any element-agent is obviously EF1 with respect to~$S_2$. For the other pairs removing agent~$S_2$ from~$\pi_2$ suffices to satisfy the EF1 condition. Thus the EF1 condition is satisfied for the element-agents in~$\pi_1$. 
    
    \item For the guard-agent~$g_1\in\pi_1$ it holds that~$\util_{g_1}(\pi_1) = B$ and~$\util_{g_1}(\pi_2/\{g_2\}) = B$, thus removing agent~$S_2$ from~$\pi_2$ suffice to satisfy the EF1 condition. Additionally~$\util_{g_1}(\pi_2/\{S_2\}) = 0$, meaning that the EF1 condition is trivially satisfied. Finally for any element-agent~$a_j\in\pi_2$,~$\util_{g_1}(\pi_2/\{a_j\}) = B$, where again removing~$S_2$ guarantees EF1.
    
    \item For agent~$S_1\in\pi_1$ it holds that~$\util_{S_1}(\pi_1) = 1+\sum_{i\in I}s_i = 1+B$, while~$\util_{S_1}(\pi_2/\{S_2\}) = 1+B$, thus removing any agent from~$\pi_2$ will satisfy EF1. For any element-agent~$a_j\in\pi_2$,~$\util_{S_1}(\pi_2/\{a_j\}) = 1 + 2B - s_j$ and for the guard-agent~$g_2$,~$\util_{S_1}(\pi_2/\{g_2\}) = 2B$. In both cases removing~$S_2$ guarantees EF1.
    \end{itemize}
    The case from the point of view of partition~$\pi_2$ is symmetrical.
    
    For the opposite direction, let partition~$\pi = (\pi_1, \pi_2)$ be EF1. Note that \Cref{Claim1,Claim2,thm:setAgents:sameUtils} also hold for EF1 partitions since they hold for EFX allocations. Without loss of generality let agents~$S_1$ and~$g_1$ belong in~$\pi_1$ and agents~$S_2$ and~$g_2$ belong in~$\pi_2$ and let~$I = \{i\in [2N] \mid a_i\in \pi_1\}$. We will show that~$\sum_{i\in I} s_i = \sum_{j\in[2N]/I} s_j$. Since~$\pi$ is EF1, it must hold that~$\util_{S_1}(\pi_1) \geq \util_{S_1}(\pi_2/\{S_2\})-\util_{S_1}(g_2)$. This holds because agent~$S_1$'s valuation for~$g_2$ is 1 and~$S_1$ does not value any other element in~$\pi_2$ higher than 1. Thus~$\util_{S_1}(\pi_1) \geq \util_{S_1}(\pi_2/\{S_2\})-\util_{S_1}(g_2)$, which means that~$1+ \sum_{i\in I}s_i \geq \sum_{j\in [2N]/I}s_j$. Symmetrically and by following the same reasoning we get that~$1 + \sum_{j\in [2N]/I} s_j \geq \sum_{i\in I} s_i$. Doing the calculation and following similar reasoning as in \Cref{thm:setAgents:sameUtils} we get that~$|\sum_{i\in I}s_i - \sum_{j\in[2N]/I}s_j |\leq 1$, where applying the same reasoning as in the proof of \Cref{thm:setAgents:sameUtils} it must be the case that~$\sum_{i\in I}s_i = \sum_{j\in[2N]/I}s_j$, which constitutes~$I$ a solution for instance~$S$.
\end{proof}

Maybe surprisingly, our above hardness results are tight. Specifically, if we restrict the friendship graph even more, we show that a fair partition is guaranteed to exist for some of the studied fairness notions. For those without such a guarantee, we show a simple condition that allows us to decide such instances in linear time. The result also contrasts \Cref{thm:MMS:NPh:verify} in the sense that under this restriction, verifying whether a given partition is MMS is computationally hard but, at the same time, it is easy to find one.

\begin{proposition}\label{thm:vc1graphs:polytime}
	If the friendship graph is a graph of vertex cover number~$1$, a~$\phi$ partition is guaranteed to exist and can be found in linear time for every~$\phi\in\{\text{EFX},\text{EF$1$},\text{MMS}\}$. Under the same restriction, the existence of EF, EFX$_0$, or PROP partition can be decided in linear time.
\end{proposition}
\begin{proof}
	The vertex cover number~$1$ graph is a disjoint union of a star and of an edgeless graph. Let us start with the first part of the statement, that is, let~$\phi\in\{\text{EFX},\text{EF$1$},\text{MMS}\}$. We create a partition~$\pttn$ by allocating the center~$c$ of the star together with~$\lceil \pttn/\parts \rceil - 1$ most valuable leaves of the star to the part~$\pttn_1$ (if there are not enough leaves, we fill the part with arbitrary isolated agents). The remaining agents are partitioned in an arbitrary way. Now, the agents in~$\pttn_1$ are clearly not envious. Let~$a$ be a leaf not in~$\pttn_1$. This leaf is clearly envious towards any agent~$\pttn_1 \setminus \{c\}$. However, the envy can be eliminated by removing~$c$ -- for EF1 we can choose~$c$ by definition; in EFX,~$c$ is the only friend of~$v$; in MMS, the MMS-share of~$a$ is anyway zero, so we do not need to care about leaves at all.
	
	For the second part of the statement, w.l.o.g. assume that~$\lfloor n/\parts \rfloor \geq 3$. First, we observe that every friend~$a$ of the star center~$c$,~$a$ must be in the same part as~$c$. If this is not the case, the utility of~$a$ is~$0$,~$\operatorname{PROP-share}(a) = \util_a(c) / \parts > 0$, and~$a$ can increase its utility by joining~$c$; since the parts are of size at least~$3$ and~$c$ is the only friend of~$a$, this holds even if we remove any non-friend of~$a$ from~$c$'s part. Therefore, unless all leaves of the star may be assigned to the same part as the star center~$c$, the response is always \emph{no}. Otherwise, we return \emph{yes}, since an arbitrary partition of this form is fair.
\end{proof}

As our previous results clearly indicate, if we allow for arbitrary utilities, then, under the standard theoretical assumptions, there is no hope for tractable algorithms even for graphs of constant vertex cover number. In the following, we show that if we additionally restrict the utilities, the situation is much more positive.

The proof of \Cref{thm:binaryFPTbyVC} is rather technical and involves two levels of guessing combined with ILP. First, it guesses what the solution looks like on the vertices of a vertex cover. Then, it guesses the neighborhood structure and tries to
verify the fairness for the independent vertices.

\begin{theorem}\label{thm:binaryFPTbyVC}
	If the utilities are binary, then for every notion of fairness~$\phi\in\{\text{EF},\text{EFX}_0,\text{EFX},\text{EF1},\text{PROP},\text{MMS}\}$, there is an \FPT algorithm parameterized by the vertex cover number~$\vc$ that decides whether a partition which is fair with respect to~$\phi$ exists.
\end{theorem}
\begin{proof}
    Let~$G = (V, E)$ and let~$U$ be the given vertex cover of~$G$.
    Note that we are mainly concerned with the groups that occupy the vertices in~$U$.
    Recall that the vertices in~$I = V \setminus U$ only have their neighborhood in~$U$.
    Therefore, we divide the vertices in~$I$ into types -- two vertices are of the same type if they have the same neighbors (in~$U$).
    For a set~$X \subseteq U$, we write~$I^X$ for the set~$\{ v \in I \mid N(v) = X \}$.
    Note that the number of types is bounded by~$2^{\vc}$.
    We denote by~$n^X$ the size of~$I^X$.

    We begin by presenting the core ideas of the proof for~$\phi \in \{ \text{PROP},\text{MMS} \}$.
    Observe that in the case of binary utilities, we can compute the~$\phi$-share~$\phi(v)$ for every vertex~$v \in V$.
    If~$\phi = \text{MMS}$, then~$\phi(v) = \lfloor \deg(v) / \parts \rfloor$.
    If~$\phi = \text{PROP}$, then~$\phi(v) = \lceil \deg(v) / \parts \rceil$.
    
    Our algorithm starts with a guessing phase (by exhaustive search through all possibilities).
    We guess the number~$\ell$ of the main groups in the vertex cover and a partition of the vertices in~$U$ to~$\ell$ groups.
    We denote by~$U_g$ the vertices assigned to a group~$g \in [\ell]$.
    Since~$\ell < \vc$, this takes~$O(\vc \cdot \vc^{\vc})$ time.
    As we now know how many neighbors of the same color the vertices in~$U$ shall have, we set~$\phi_g(v) = \phi(v) - |U_g \cap N(v)|$.
    If~$\parts > \ell$, then we set~$\hat{\ell} = \ell + 1$ and set~$\hat{\ell} = \ell$ otherwise.
    The possible ``extra'' color is there to represent groups that do not contain any of the~$\ell$ colors used in the vertex cover.
    Next, we guess for each~$X$ the set of~$\hat{\ell}$ groups that may contain vertices in~$I^X$.
    We shall seek a partition that for each group contains at least one vertex of type~$I^X$ if and only if its color is guessed for~$I^X$.
    Note that this takes~$(2^{\vc})^{2^{\hat{\ell}}} \subseteq 2^{2^{O(\vc)}}$ time.
    Based on this guess, we observe that we can determine whether the vertices in~$I$ meet their~$\phi$-share.
    If this is not the case, we reject the current guess and continue with the next one.
    Otherwise, we use an ILP formulation to decide if one can find a realization of the above guesses that fulfills the~$\phi$-share of vertices in~$U$ as follows.
    
    We use integer variables~$x^X_g$ that represent the number of vertices in~$I^X$ assigned to a group~$g \in [\hat{\ell}]$.
    For~$x^X_g$ to comply with the guess, we set~$x^X_g = 0$ if the group~$g$ has no vertices assigned in~$I_X$ and otherwise require~$x^X_g \ge 1$.
    We require
    \begin{align}
        \sum_{X \subset U, v \in X} x^X_g &\ge \phi_g(v)  &\forall v \in U_g, \forall g \in [\ell] \label{eq:binaryFPTbyVC:ShareILP:VCShare}\\
        \sum_{g \in [\hat{\ell}]} x^X_g &= n^X  &\forall X \subseteq U \label{eq:binaryFPTbyVC:ShareILP:ISPartition} \\
        \lfloor n/k \rfloor \le |U_g| + \sum_{X \subset U} x^X_g &\le \lceil n/k \rceil &\forall g \in [\ell] \label{eq:binaryFPTbyVC:ShareILP:balanced}
    \end{align}
    If~$\hat{\ell} > \ell$, then we further add condition
    \begin{equation}\label{eq:binaryFPTbyVC:ShareILP:ExtraColorsDivisible}
        (\parts - \ell) \lfloor n/k \rfloor \le \sum_{X \subseteq U} x^X_{\ell+1} \le (\parts - \ell) \lceil n/k \rceil
    \end{equation}
    
    If the ILP has a solution, then we have a sought partition that complies with all the guesses.
    To see this, note that by the lower and upper bounds on variables~$x^X_g$ the numbers of vertices must comply with the guess.
    Furthermore, by \eqref{eq:binaryFPTbyVC:ShareILP:ISPartition} it is possible to partition~$I^X$ into groups~$I^X_g$ of size~$x^X_g$.
    Let~$\pi_g = U_g \cup \bigcup_{X \subseteq U} I^X_g$ be a group for~$g \in [\ell]$.
    By \eqref{eq:binaryFPTbyVC:ShareILP:VCShare} the~$\phi$-share of vertices in~$U$ is met.
    The~$\phi$-share of vertices in~$I$ is met since the partition complies with the guess.
    Finally, if~$\hat{\ell} > \ell$, conditions \eqref{eq:binaryFPTbyVC:ShareILP:ExtraColorsDivisible} ensure that it is possible to divide the vertices assigned to the group~$\ell+1$ into~$\parts - \ell$ groups.
    We can do this in an arbitrary way, since the~$\phi$-share of these vertices is met.
    This, together with \eqref{eq:binaryFPTbyVC:ShareILP:balanced} implies that we have found the sought balanced partition~$\pi$.
    
    For the other direction, we prove that if the given instance admits a solution, then there exists an initial guess for which the ILP is feasible.
    Suppose that~$\pi$ is a balanced~$\parts$-partition satisfying~$\phi$.
    Without loss of generality, we assume that~$\pi$ assigns vertices in~$U$ to groups in~$[\ell]$ for some~$\ell \le \vc$.
    We claim that~$\pi$ complies with a guess that assigns groups in~$[\ell]$ to the vertices in~$I^X$ if~$I^X \cap \pi_g \neq \emptyset$.
    If~$\parts > \ell$, this guess additionally contains nonempty groups~$I^X$ if~$I^X \cap \cup_{g > \ell} \pi_g \neq \emptyset$.
    We let~$x^X_g = |I^X \cap \pi_g|$ for~$g \in [\ell]$ and if~$\parts > \ell$, we set~$x^X_{\ell+1} = |I^X \cap \cup_{g > \ell} \pi_g|$.
    It is straightforward to verify that such values of variables~$x$ verify conditions \eqref{eq:binaryFPTbyVC:ShareILP:VCShare}--\eqref{eq:binaryFPTbyVC:ShareILP:ExtraColorsDivisible}.

    For envy-based~$\phi\in\{\text{EF},\text{EFX}_0,\text{EFX},\text{EF1}\}$, we proceed in a similar way.
    The few exceptions are that we cannot compute the values~$\phi(v)$ for~$v \in U$ in advance and instead verify these on the fly using the ILP.
    Again, we start by guessing the number~$\ell$ and the groups~$U_g \subseteq U$ for~$g \in [\ell]$.
    If~$\parts > \ell$, then we set~$\hat{\ell} = \ell + 1$ and set~$\hat{\ell} = \ell$ otherwise.
    We guess for each~$X$ the set of~$\hat{\ell}$ groups that may contain vertices in~$I^X$.
    We shall seek a partition that for each group contains at least one vertex of type~$I^X$ if and only if its color is guessed for~$I^X$.
    Based on the two guesses, we verify~$\phi$ for vertices in~$I$:
    \begin{description}
        \item[EF] Each vertex in~$I^X$ in the group~$g \in [\ell]$ must have at least as many neighbors in~$U_g$ as in other colors. Moreover, the group~$\ell+1$ can only be presented in the set~$I^{\emptyset}$.
        \item[EFX$_0$] It is similar to EF.
        \item[$\text{EFX}=\text{EF1}$] Each vertex in~$I^X$ in the group~$g \in [\ell]$ must have at least as many neighbors in~$U_g$ as in other colors minus~1. The group~$\ell+1$ can only be presented if it has in each group in~$[\ell]$ at most one neighbor.
    \end{description}
    Note that this ensures~$\phi$ for all vertices in~$I$.
    In the ILP model, we reuse conditions \eqref{eq:binaryFPTbyVC:ShareILP:ISPartition} and \eqref{eq:binaryFPTbyVC:ShareILP:balanced} as well as the bounds on the integer variables~$x^X_g$; ensuring the balancedness of the partition and that every vertex in~$I^X$ is in exactly one group.
    If~$\hat{\ell} > \ell$, then we further add the condition \eqref{eq:binaryFPTbyVC:ShareILP:ExtraColorsDivisible} and introduce additional integer variables~$y^u$ with~$y^u \ge 0$ and the condition for all~$u \in U$
    \begin{equation}\label{eq:binaryFPTbyVC:EFILP:yForExtraColors}
        (\parts - \ell) y^u \le -1 + \sum_{X \subseteq U, u \in X} x^X_{\ell+1} \le (\parts - \ell) \cdot (y^u + 1) \,.
    \end{equation}
    It is straightforward to verify that \eqref{eq:binaryFPTbyVC:EFILP:yForExtraColors} ensure that \[y^u+1 = \left\lceil \sum_{X \subseteq U, u \in X} x^X_{\ell+1} / (\parts - \ell) \right\rceil \,.\]
    That is,~$y^u+1$ is the number of vertices assigned to groups~$\ell+1, \ldots, \parts$ if we assign those in the round-robin fashion.
    Now, we are ready to verify~$\phi$ for vertices in~$U$.
    For~$\phi=\text{EF}$ we have
    \begin{equation}\label{eq:binaryFPTbyVC:EFILP:EFforMainGroups}
        \begin{split}
            |N(u) \cap U_g| + \sum_{X \subseteq U, g \in X} x^X_g \ge&\\
            |N(u) \cap U_{g'}| - \nu(u,g') + \sum_{X \subseteq U, g \in X} x^X_{g'} & \\
            \forall u \in U_g \forall g \neq g' \in [\ell] &
        \end{split}
    \end{equation}
    Note that the left-hand side is the utility of~$u$ in the group~$g$.
    We set the constant~$\nu(u,g') \in \{ 0,1 \}$ to~$1$ if and only if all vertices assigned to color~$g'$ are neighbors of~$u$.
    Note that it is possible to compute~$\nu(u,g')$ for all~$u \in U$ based on the initial guesses.
    Now, the right-hand side is exactly the utility of~$u$ if it moves from its group~$g$ to the group~$g'$.
    Finally, if~$\hat{\ell} > \ell$, we add the conditions
    \begin{equation}\label{eq:binaryFPTbyVC:EFILP:EFforOthers}
        |N(u) \cap U_g| + \sum_{X \subseteq U, g \in X} x^X_g \ge y^{u} \quad \forall u \in U
    \end{equation}

    We again show that the provided ILP is a valid model.
    Suppose there is a solution~$x$ to the ILP with conditions \eqref{eq:binaryFPTbyVC:ShareILP:ISPartition}, \eqref{eq:binaryFPTbyVC:ShareILP:balanced}, \eqref{eq:binaryFPTbyVC:ShareILP:ExtraColorsDivisible} and \eqref{eq:binaryFPTbyVC:EFILP:yForExtraColors}--\eqref{eq:binaryFPTbyVC:EFILP:EFforOthers} and the box constraints based on the initial guesses.
    We assign groups~$g \in [\ell]$ as before, that is, we partition the vertices in~$I^X$ so that~$\pi_g$ contains unique~$x^X_g$ such vertices.
    We assign groups~$\ell+1, \ldots, \parts$ in a round-robin fashion in arbitrary order of sets~$I^X$ using always exactly~$x^X_{\ell+1}$ vertices of each type~$X$.
    Note that by \eqref{eq:binaryFPTbyVC:ShareILP:ISPartition} this is a valid partition of~$V$ and that it is balanced due to \eqref{eq:binaryFPTbyVC:ShareILP:balanced} and \eqref{eq:binaryFPTbyVC:ShareILP:ExtraColorsDivisible}.
    Now, the vertices in~$I$ do not envy as otherwise the initial guess was rejected.
    The vertices in~$U$ do not envy other groups~$g' \in [\ell]$ by \eqref{eq:binaryFPTbyVC:EFILP:EFforMainGroups}.
    It is not difficult to verify that in~$\pi_{g'}$ with~$g' \in \{ \ell+1, \ldots, \parts \}$ contain at either~$y^u$ or~$y^u-1$ vertices in the neighborhood of~$u \in U$.
    Moreover, using round-robin yields that the maximum occupancy of such a group in the neighborhood of~$u$ is exactly~$y^u$.
    Note that this is secured by \eqref{eq:binaryFPTbyVC:EFILP:yForExtraColors}.
    Then~$u$ does not envy any such group, as is established by \eqref{eq:binaryFPTbyVC:EFILP:EFforOthers}.
    
    In the other direction, suppose that the given instance admits an EF~$\parts$-balanced partition~$\pi$.
    W.l.o.g.~$\pi$ uses first~$\ell \le |U|$ groups on the vertex cover~$U$.
    We set variables~$x^X_g$ for~$g \in [\ell]$ as above, i.e.,~$x^X_g = | \{ v \in I^X \mid v \in \pi_g \}|$.
    We set the variables~$x^X_{\ell+1}$ to~$|I^X| - \sum_{g \in [\ell]} x^X_g$.
    We claim that for some initial guess of the solution pattern, this is a valid solution to the resulting ILP above.
    First of all, since there is no envy for vertices in~$I$ and based on the above setting, we know for every~$g \in [\hat{\ell}]$ in which~$I^X$ there are some members of the group -- this allows us to reverse engineer the initial guess (note that it is unique and was not rejected).
    Conditions \eqref{eq:binaryFPTbyVC:ShareILP:ISPartition}, \eqref{eq:binaryFPTbyVC:ShareILP:balanced}, and \eqref{eq:binaryFPTbyVC:ShareILP:ExtraColorsDivisible} are clearly satisfied since~$\pi$ was a~$\parts$-balanced partition.
    Conditions \eqref{eq:binaryFPTbyVC:EFILP:yForExtraColors} always have solution~$y^u$.
    Next, we observe that
    \[
        y^u + 1 \le \max_{i \in \{\ell + 1, \ldots, \parts \}} |\pi_i \cap N(u)| 
    \]
    holds.
    Clearly, round-robin can only lower this number as it is the most uniform, achieving the lowest possible right-hand side.
    Since~$\pi$ is~$\phi$-free for the groups in~$[\ell]$, conditions \eqref{eq:binaryFPTbyVC:EFILP:EFforMainGroups} are met; they compute the utilities.
    Similarly, for \eqref{eq:binaryFPTbyVC:EFILP:EFforOthers} holds as~$y^u$ is below of what is possibly set by~$\pi$.
    Thus, the ILP has a solution.

    It is not hard to verify (using the proof above) that one can get the ILP for~$\phi \in \{ \text{EFX}_0,\text{EFX},\text{EF1} \}$ by slight modifications in the conditions \eqref{eq:binaryFPTbyVC:EFILP:EFforMainGroups} and \eqref{eq:binaryFPTbyVC:EFILP:EFforOthers}.
    The right-hand side receives penalty of~1 (based on the initial guess).

    It remains to analyze the total time needed for the computation.
    We have already observed that the initial guessing only requires time~$2^{2^\Oh{\vc}}$.
    The number of variables of the ILPs can be bounded by a function that solely depends on the selected parameter~$\vc$.
    Therefore, the theorem follows using the celebrated result of Lentra~\cite{Lenstra83,Dadush12}.
\end{proof}

The last result of this section is an algorithm \FPT for parameterization by the number of agents. Although the algorithm is not complicated, it highlights the difference between our model and the setting of standard fair division of indivisible items, where the existence of fair outcomes may be algorithmically challenging already for a constant number of agents~\cite{BergerCFF2022,GhosalPNV2023,ChaudhuryGM2024,DeligkasEKS2024}.

\begin{observation}
	For every fairness notion~$\phi\in\{\text{EF},\text{EFX}_0,\text{EFX},\text{EF1},\text{PROP},\text{MMS}\}$, the problem of deciding whether a~$\phi$-fair partition exists is in \FPT when parameterized by the number of agents~$n$.
\end{observation}
\begin{proof}
	It is easy to see that we can assume that~$\parts < n$. Otherwise, we just partition agents into singleton parts and in polynomial time, we decide whether the partition is fair (which can be done also for MMS due to \Cref{thm:MMS:exists:if:parts:greater:maxDeg}). Therefore, there are at~$n^\Oh{n}$ different partition of the agents and for every
	~$\phi\in\{\text{EF},\text{EFX}_0,\text{EFX},\text{EF1},\text{PROP}\}$ we can verify their fairness in polynomial time. That is, if we found at least one fair partition, we return \emph{yes}, otherwise, we return~\emph{no}. For MMS, we first try all partitions of friends of every agent~$a\in\agents$ into~$\parts$ parts in order to determine its MMS-share. Once MMS-share of every agent is known, we can do the same procedure as for the other fairness notions and verifying fairness becomes easy task also for MMS. Even with this preprocessing, the running time is still~$n^\Oh{n}$, which is clearly in \FPT with respect to~$n$.
\end{proof}

\section{Tree-Like Friendship Graphs}\label{sec:trees}

As is clear from the results of the previous section, without an additional restriction of the input instances, we are not able to guarantee any existential result or positive algorithm. Therefore, in this section, we focus on the natural restriction of the underlying friendship graph; specifically, we focus on situations where the social friendship graph is similar to a tree. It should be noted that such restriction is heavily studied in the related literature; see, e.g.~\cite{BouveretCEIP2017,FarhadiGHLPSSY2019,HanakaL2022,Xiao0H2023,HosseiniNW2024}.

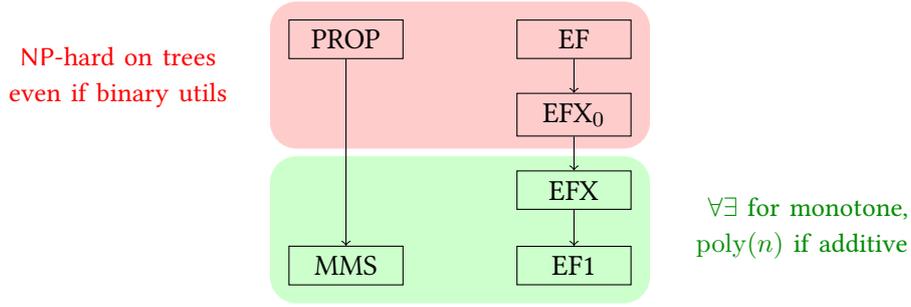
\begin{figure}[bt!]
	\centering
    \begin{tikzpicture}[every node/.style={minimum width=1.5cm}]
		\node[draw] (EF) at (1.5,0) {EF};
		\node[draw] (EFX0) at (1.5,-1) {EFX$_0$};
		\node[draw] (EFX) at (1.5,-2) {EFX};
		\node[draw] (EF1) at (1.5,-3) {EF1};
		\node[draw] (PROP) at (-1.5,0) {PROP};
		\node (PROPdummy) at (-1.5,-1) {};
		\node[draw] (MMS) at (-1.5,-3) {MMS};
		\node[draw,transparent] (MMSdummy) at (-1.5,-2) {};
		
		\draw (EF) edge[->] (EFX0);
		\draw (EFX0) edge[->] (EFX);
		\draw[->] (PROP) edge (MMS);
		\draw[->] (EFX) edge (EF1);
		
		\node[red,text width=3cm,align=center,yshift=-0.5cm] at (-4.5,0) {\small \NP-hard on trees even if binary utils};
		\node[green!55!black,text width=3cm,align=center,yshift=-0.5cm] at (4.5,-2) {\small~$\forall\exists$ for monotone,\\$\operatorname{poly}(n)$ if additive};
		
		\begin{pgfonlayer}{background}
			\fill[red!40,opacity=0.5,rounded corners=10]
				(-2.5,0.5) rectangle (2.5,-1.45);
			\fill[green!40,opacity=0.5,rounded corners=10]
				(-2.5,-3.5) rectangle (2.5,-1.55); %
		\end{pgfonlayer}
	\end{tikzpicture}
	\caption{A basic overview of our complexity results for trees.}
	\label{fig:trees-results-overview}
\end{figure}

In \Cref{thm:EF:NPh}, we showed that the situation is, from the computational complexity perspective, hopeless if we require EF, EFX$_0$, or PROP partition and the utilities can be arbitrary, even if the friendship graph is a path. On the other hand, if the utilities are binary, we can decide the existence of fair partitions in polynomial time.

\begin{proposition}\label{thm:unweighted:EF:path:P}\label{thm:unweighted:PROP:path:P}
    If the utilities are binary and~$G$ is a path, there is a linear-time algorithm that decides whether an EF or PROP partition exists. In the same setting, EFX$_0$ partitions are guaranteed to exist and can be found in linear time.
\end{proposition}
\begin{proof}
	We show the correctness of the proposition separately for each of the fairness notions of our interest. Let~$\parts$ be fixed and the fairness notion we are interested in be PROP. First, observe that for every degree-2 agent~$a\in\agents$,~$\operatorname{PROP-share}(a) = 2/\parts$, while for degree-1 agents, it is~$1/\parts$. In any case, the PROP-share is strictly greater than zero. Consequently, if there exists a part of size~$1$, we directly return \emph{no}, as the utility of an agent assigned to this part is zero and hence, no partition can be PROP. Therefore, let~$\lfloor n/\parts \rfloor \geq 2$. In this case, we create a partition~$\pttn=(\pttn_1,\ldots,\pttn_\parts)$ such that each part consists of a connected sub-path. Now, every agent is in the same part as at least one of its friends, so~$\pttn$ is clearly PROP.
	
	For EF, suppose first that all parts are of size~$1$. Then, we arbitrarily partition the agents into~$n$ parts. Even though the utility of every agent is zero, they cannot improve by replacing an agent in any of the other parts. Hence, we return \emph{yes}. Next, suppose that the parts are of size either~$1$ or~$2$ and let~$a$ be an agent assigned to a part of size~$1$ with a friend, say~$b$, assigned to a size-$2$ part. The utility of agent~$a$ in its current part is obviously~$0$, and the~$a$'s utility in part~$\{a,b\}$ is one. Therefore, there cannot be an EF partition, so in this case, we return \emph{no}. Finally, if all parts are of size at least two, the situation is the same as in the case of PROP; we can find a partition such that every agent is in the same part as at least one of its friends. Since each agent has at most two friends, the partition is clearly EF. 
	
	Finally, as EF implies EFX$_0$, we need to separately examine only the case where some parts are of size~$1$ and some of them are of size~$2$. In the case of EF, the problematic part was that an agent in size-$1$ part could join its friend in size-$2$ part. However, EFX$_0$ states that any agent we remove from the target part eliminates the envy. Since there is only the friend~$b$, the envy is eliminated as~$a$'s utility is still zero. Therefore, EFX$_0$ partitions are guaranteed to exist in this case.
\end{proof}

Then, a natural question arises. Can we generalize the result from \Cref{thm:unweighted:EF:path:P} to a more general class of graphs? In the following result, we show that this is not the case. More specifically, we show that it is computationally hard to decide whether EF or PROP partition already exists for trees of constant depth.

\begin{theorem}
    \label{thm:unweighted:EF:tree:NPh}\label{thm:unweighted:PROP:tree:NPh}
	It is \NPc and \Wh when parameterized by the number of parts~$\parts$ to decide whether PROP, EF, or EFX$_0$ partition exists, even if the utilities are binary and~$G$ is a tree of depth~$2$.
\end{theorem}
\begin{proof}
	To prove the theorem, we give a parameterized reduction from the \probName{Unary Bin Packing}, where we are given a multiset~$S = \{s_1,\ldots,s_N\}$ of items sizes, a number of bins~$B$, and a capacity~$c$ of every bin. The goal is to decide whether there exists an allocation~$\alpha\colon S\to[m]$ such that for every~$j\in[B]$ it holds that~$\sum_{s\colon \alpha(s) = j} s \leq c$. \probName{Unary Bin Packing} is known to be \Wh when parameterized by the number of bins, even if the size of each item is encoded in unary, all~$s_i\in S$ are at least~$2$,~$c\geq4$, and~$\sum_{i=1}^N s_i = c\cdot B$~\cite{JansenKMS2013}. 
	
	The idea behind the construction is to create a gadget for each item~$s_i\in S$ such that the whole gadget is part of the same part in every solution. For connectivity reasons, we add one connecting gadget that connects all the item gadgets together and, in every solution, forms one extra part. Formally, given an instance~$\mathcal{I} = (S,B,c)$ of the \probName{Unary Bin Packing} problem, we create an equivalent instance of our problem as follows. For every item~$s_i\in S$, we create an \emph{item-gadget}~$X_i$, which is a star with center~$c_i$ and~$s_i - 1$ leaves~$v^i_1,\ldots,v^i_{s_i - 1}$. Additionally, we add one \emph{connecting-gadget}~$X_g$, which is again a star with center~$c^*$ and~$c-1$ leaves~$v^*_1,\ldots,v^*_{c-1}$. Moreover, there is an edge~$\{c_i,c^*\}$ for every~$i\in[N]$. Overall, there are~$c\cdot B + c$ agents. To finalize the construction, we set~$\parts = B + 1$, and we recall that the utilities are binary. For an illustration of the construction, we refer the reader to \Cref{fig:thm:unweighted:EF:tree:NPh}.
	
	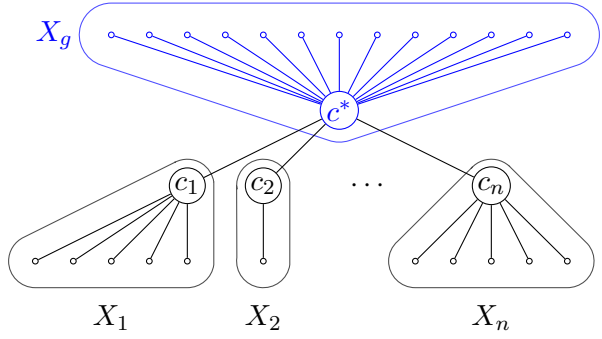
\begin{figure}[bt!]
		\centering
		\begin{tikzpicture}
        
			\node[draw,circle,inner sep=0.75pt,blue] (cg) at (0,0) {$c^*$};
			\foreach[count=\i] \x in {-3,-2.5,...,3}{
				\node[draw,circle,inner sep=0.75pt,blue] (cg\i) at (\x,1) {};
				\draw[blue] (cg) -- (cg\i);
			}
			\begin{pgfonlayer}{background}
				\draw[blue!70] \convexpath{cg,cg1,cg13}{12pt};
				\node[blue] at (-3.75,1) {$X_g$};
			\end{pgfonlayer}
			
			\node[draw,circle,inner sep=0.75pt] (c1) at (-2,-1) {$c_1$};
			\foreach[count=\i] \x in {-2,-2.5,...,-3.5,-4}{
				\node[draw,circle,inner sep=0.75pt,black] (c1\i) at (\x,-2) {};
				\draw (c1) -- (c1\i);
			}
			\begin{pgfonlayer}{background}
				\draw[black!70] \convexpath{c1,c11,c15}{10pt};
				\node[] at (-3,-2.75) {$X_1$};
			\end{pgfonlayer}
			
			\node[draw,circle,inner sep=0.75pt] (c2) at (-1,-1) {$c_2$};
			\foreach[count=\i] \x in {-1}{
				\node[draw,circle,inner sep=0.75pt,black] (c2\i) at (\x,-2) {};
				\draw (c2) -- (c2\i);
			}
			\begin{pgfonlayer}{background}
				\draw[black!70] \convexpath{c2,c21}{10pt};
				\node[] at (-1,-2.75) {$X_2$};
			\end{pgfonlayer}

			\node[] at (0.4,-1) {$\cdots$};
			
			\node[draw,circle,inner sep=0.75pt] (cn) at (2,-1) {$c_n$};
			\foreach[count=\i] \x in {1,1.5,...,3}{
				\node[draw,circle,inner sep=0.75pt,black] (cn\i) at (\x,-2) {};
				\draw (cn) -- (cn\i);
			}
			\begin{pgfonlayer}{background}
				\draw[black!70] \convexpath{cn,cn5,cn1}{10pt};
				\node[] at (2,-2.75) {$X_n$};
			\end{pgfonlayer}
			
			\draw (cg) edge (c1) edge (c2) edge (cn);
		\end{tikzpicture}
		\caption{An illustration of the construction from the proof of \Cref{thm:unweighted:EF:tree:NPh}. In blue, we highlight the connecting gadget~$X_g$ with~$c-1$ leaves. Item-gadgets~$X_i$,~$i\in[N]$, are in one-to-one correspondence with the elements of~$S$.}
		\label{fig:thm:unweighted:EF:tree:NPh}
	\end{figure}
	
	For correctness, let~$\mathcal{I}$ be a \emph{yes}-instance and let~$\alpha$ be a solution allocation for~$\mathcal{I}$. We create a partition~$\pttn=(\pttn_1,\ldots,\pttn_\parts)$ such that for every~$i\in[\parts-1]$ we have~$\pttn_i = \{ V(X_j) \mid \alpha(s_j) = i \}$ and we set~$\pttn_\parts = V(X_g)$. Since~$\alpha$ is a solution allocation, all parts are of the same size. What remains to show is that~$\pttn$ is fair. Since all leaves are in the same part as all of their friends, they do not have envy towards any agent. Now, let~$c_i$ be a center of some item-gadget~$X_i$,~$i\in[N]$.~$c_i$ is in the same part as all leaves~$v^i_1,\ldots,v^i_{s_i - 1}$. Since~$s_i$ is at least two,~$c_i$ is in the same part with at least one friend. The only other part containing a friend of~$c_i$ is~$\pttn_\parts$, which contains~$c^*$. However,~$c_i$'s utility is at least as large as if it would replace some leave~$v^*_j$,~$j\in[c-1]$, in~$\pttn_\parts$. That is,~$\pttn$ is EF from the viewpoint of all item-gadgets' centers. For PROP, observe that~$\operatorname{PROP-share}(c) = s_i / \parts$, which is strictly smaller than~$s_i - 1$. Finally, we focus on connecting-gadget~$X_g$. By the same argumentation as in the case of item-gadgets,~$\pttn$ is fair for all leaves as they are in the same part as all of their friends. Hence, it remains to show that~$\pttn$ is fair for~$c^*$. Observe that~$\util_{c^*}(\pttn) = |\pttn_\parts| - 1 = c - 1$, as the whole part the agent~$c^*$ belongs to consists of its friends. Since all~$s_i$,~$i\in[N]$, are at least two, for every part~$\pttn_j$,~$j\in[\parts-1]$, we have~$\pttn_j \cap F(c^*) \leq c/2$. Hence, the agent~$c^*$ cannot improve by moving to another part, and~$\pttn$ is, therefore, EF. Next, it holds that~$\operatorname{PROP-share}(c^*) \leq (B\cdot (c/2) + c)/\parts = ((\parts-1)\cdot (c/2) + c)/\parts = ((\parts\cdot (c/2))/\parts - (c/2)/\parts + c/\parts = c/2 - c/2\parts + c/\parts = c/2 + (-c+2c)/2\parts = c/2 + c/2\parts < \util_{c^*}(\pttn)$. That is,~$\pttn$ is PROP-fair also for~$c^*$. Consequently,~$\pttn$ is EF and PROP.
	
	In the opposite direction, let~$\mathcal{J}$ be a \emph{yes}-instance and let~$\pttn$ be EF, EFX$_0$ or PROP partition. First, we show that the whole connecting-gadget forms its own part. Assume that this is not the case and let~$v^*_i$,~$i\in[c-1]$, be a leaf agent which is in a different part than the agent~$c^*$ and, without loss of generality, let~$c^*\in\pttn_\parts$. As~$F(v^*_i) = \{c^*\}$, we have~$\util_{v^*_i}(\pttn) = 0$, so agent is envious towards any non-friend in~$\pttn_\parts$. Moreover, we assumed that~$c > 4$, so there are at least three non-friends of~$v^*_i$ in~$\pttn_\parts$ and the envy cannot be eliminated by removal of any of them. Hence, we have a contradiction that~$\pttn$ is EFX$_0$. Also,~$\operatorname{PROP-share}(v^*_i) = 1/\parts > 0$. This contradicts that~$\pttn$ is EF, EFX$_0$, and PROP. Hence, all agents of~$X_g$ are in the same part. Since~$|V(X_g)| = c$, their part is actually equal to~$V(X_g)$ and, without loss of generality, for the rest of the proof, we will assume that~$\pttn_\parts = V(X_g)$. Now, we show a similar property for item-gadgets.
	
	\begin{claim}\label{thm:unweighted:EF:tree:NPh:itemGadgetsSubsetsParts}
		For every~$X_i$,~$i\in[N]$, there exist exactly one part~$\pttn_j$,~$j\in[\parts-1]$, such that~$V(X_i) \cap \pttn_j \not= \emptyset$. Consequently,~$V(X_i) \subseteq \pttn_j$.
	\end{claim}
	\begin{claimproof}
		For the sake of contradiction, assume that there is an item-gadget~$X_i$,~$i\in[N]$, such that there is at least one leaf that is not in the same part as the center~$c_i$. Without loss of generality, let the lead be~$v^i_1$. Since~$c_i$ is the only friend of~$v^i_1$, its utility is~$0$. However, if~$v^i_1$ replaces any agent except for~$c_i$ in~$c_i$'s part, its utility increases to~$1$. This contradicts that~$\pttn$ is EF. Moreover, as~$c>4$, this envy cannot be eliminated by the removal of any agent from~$c_i$'s part, so we have a contradiction that~$\pttn$ is EFX$_0$. Similarly,~$\operatorname{PROP-share}(v^i_1) = 1/\parts > 0$, which contradicts that the partition is PROP. Therefore, the whole item-gadget is necessarily in the same part.
	\end{claimproof}
	
	By \Cref{thm:unweighted:EF:tree:NPh:itemGadgetsSubsetsParts}, we get that all vertices of every item-gadgets are in the same part. Using this, we create an allocation~$\alpha$ as follows: for every~$s_i\in S$, we set~$\alpha(s_i) = j$, where~$j$ is such that~$V(X_i) \subseteq j$. Since such~$j$ is unique, the allocation is well-defined. For the sake of contradiction, assume that there exists a bin~$j\in[B]$ such that~$\sum_{s\colon \alpha(s) = j} s > c$. Then necessarily~$|\pttn_j| > c$, which contradicts that~$\pttn$ is balanced. Therefore, all items are allocated, and no bin's capacity is exceeded. Consequently,~$\alpha$ is a solution allocation; therefore,~$\mathcal{I}$ is also a \emph{yes}-instance.
	
	Finally, the construction preserves that~$\parts = B+1$, and since \probName{Unary Bin Packing} is \Wh when parameterized by~$B$, we obtain that our problem is \Wh with respect to~$\parts$. Moreover, our parameterized reduction runs in polynomial time. Hence, it is also a polynomial reduction, so the problems of our interest are also \NPc. 
\end{proof}

Observe that the tree constructed in the proof of \Cref{thm:unweighted:EF:tree:NPh} is very shallow. In particular, it can be shown that these trees have constant treedepth ($3$, to be precise), which shows a strong intractability result for this particular structural restriction and, therefore, also for more general parameters such as the celebrated treewidth of the friendship graph.

Our next result complements the lower-bound given in \Cref{thm:unweighted:EF:tree:NPh}. Specifically, we show that if~$G$ is a tree, then there exists an \XP algorithm for parameterization by the number of parts~$\parts$ for every notion of stability assumed in this work. In other words, for every constant~$\parts$, the problem can be solved in polynomial time on trees. Our algorithm is even stronger and works for any class of graphs of constant tree-width.

\begin{theorem}
\label{thm:tw:XP}
	If the utilities are binary and~$G$ is a graph of tree-width~$\tw$, then for every fairness notion~$\phi\in\{\text{EF},\text{EFX$_0$},\text{EFX},\text{EF1},\text{PROP},\text{MMS}\}$, there is an \XP algorithm parameterized by~$k$ and~$\tw$ that decides whether a partition~$\pttn$, which is fair with respect to~$\phi$, exists.
\end{theorem}
\begin{proof}
	\newcommand{\DP}{\operatorname{DP}}
	The algorithm is, as is usual in similar problems, based on leaf-to-root dynamic programming along the nice tree decomposition. Recall that, if a nice tree decomposition is not given as part of the input, we can find a tree decomposition of the optimal width in time~$2^\Oh{k^2}\cdot n^\Oh{1}$~\cite{KorhonenL2023} and turn it into a nice one in linear time~\cite{CyganFKLMPPS2015}.
	
	The computation is formally defined as follows. In every node~$x\in V(T)$ of the nice tree decomposition, we store a dynamic programming table~$\DP_x[P,f,c,s]$, where
	\begin{itemize}
		\item~$P\colon \beta(x)\to[\parts]$ is a partition of bag agents into coalitions,
		\item~$f\colon\beta(x)\times[\parts]\to[n]$ is a function returning for each bag agent~$a$ and each part~$i$ the number of future friends of~$a$ assigned to part~$\pttn_i$, and
		\item~$c\colon\beta(x)\times[\parts]\to[n]$ is a function returning for each bag agent~$a$ and each part~$i$ the number of already forgotten neighbors of~$a$ allocated to coalition~$i$, and
		\item~$s\colon [\parts]\to[n]$ is a function assigning to each part~$i$ the number of already forgotten agents in the sub-tree rooted in~$x$ that are assigned to~$\pttn_i$.
	\end{itemize}
	We call the quadruplet~$(P,f,c,s)$ a \emph{signature}. The table~$\DP_x$, for a signature~$(P,f,c,s)$, stores \texttt{true} if and only if there exists a \emph{partial~$\parts$-partition} (not necessarily balanced)~$\pttn^x_{P,f,c,s}$ of~$V^x$ (whenever the signature~$(P,f,c,s)$ and/or the node~$x$ are clear from the context, we refer to~$\pttn^x_{P,f,c,s}$ simply as~$\pttn^x$ or~$\pttn$, respectively) such that
	\begin{enumerate}
		\item for every~$a\in\beta(x)$ we have~$\pttn(a) = P(a)$,
        \item for every~$a\in\beta(x)$ and every~$i\in[\parts]$ we have~$|(\pttn_i \cap F(a))\setminus\beta(x)| = c(a,i)$,
        \item for every~$a\in\beta(x)$ it holds that~$\sum_{i=1}^\parts c(a,i) + f(a,i) = |F(a)\setminus\beta(x)|$,
        \item for every~$i\in[\parts]$ we have~$|\pttn_i\setminus\beta(x)| = s(i)$,
        \item for every~$a\in V^x\setminus\beta(x)$ the partition~$\pttn$ is fair for~$a$ with respect to~$\phi$, and
        \item for every~$a\in \beta(x)$ the partition~$\pttn$ is fair for~$a$ with respect to~$\phi$, assuming that each part~$i\in[\parts]$ contains~$f(a,i)$ additional friends of~$a$.
	\end{enumerate}
	Otherwise, the stored value is \texttt{false}. It is easy to see that the size of~$\DP_x$ is~$\parts^\Oh{\tw} \cdot n^\Oh{\parts\cdot\tw}\cdot n^\Oh{\parts\cdot\tw} \cdot n^\Oh{\parts} = n^\Oh{\parts\cdot\tw}$, which is clearly in \XP. Once the dynamic programming table for the root node~$r$ is correctly computed, we read the solution just by checking whether~$\DP_r[ \emptyset, \emptyset, \emptyset, (\lceil n/\parts \rceil,\ldots,\lceil n/\parts \rceil, \lfloor n/\parts \rfloor,\ldots,\lfloor n/\parts \rfloor) ]$ is set to \texttt{true}, where we slightly abuse the notation and use~$\emptyset$ to denote an empty function. Observe that if it is the case, then there exists a~$\parts$-partition~$\pttn^x$ of~$V^r = V$ such that it is fair with respect to~$\phi$ for every agent~$a\in V$ and the parts are balanced.
	
	In the rest of the proof, we describe the computation of the dynamic programming table separately for each type of node of the tree decomposition. Along with the description, we also prove that the computation is correct and that for each table, we spend only~$n^\Oh{\parts\cdot\tw}$ time, which will conclude the proof.
	
	\paragraph{Leaf Node.} The bag of every leaf node~$x$ is, by definition, empty. Therefore, almost all signature elements have an empty domain, except for the size function~$s$, which must be a constant zero.
	\begin{equation*}
		\DP_x[ P, f, c, s ] = \begin{cases}
		    \texttt{false} & \text{if } \exists i \in [k]\colon s(i) \not= 0 \text{ and}\\
            \texttt{true}  & \text{otherwise.}
		\end{cases}
	\end{equation*}
	
	Now, we show that computation in leaf nodes is indeed correct.
	
	\begin{claim}\label{lem:tw:leaf:correct}
		$\DP_x[P,f,c,s] = \texttt{true}$ if and only if there exists a corresponding partial~$k$-partition~$\pttn^x$.
	\end{claim}
	\begin{claimproof}
		First, assume that~$\DP_x[P,f,c,s] = \texttt{true}$. From the definition of the computation, the size function~$s$ is a constant zero. Therefore, if we set~$\pttn^x = (\emptyset,\ldots,\emptyset)$, we indeed obtain a solution, as all parts are of the correct size and~$\bigcup_{i\in[\parts]} \pttn^x_i = \emptyset = V^x$.

        In the opposite direction, let there exist a partial~$\parts$-partition~$\pttn$ corresponding to a signature~$(P,f,c,s)$. Since~$V^x = \emptyset$, we must have~$s(i) = 0$ for every~$i\in[\parts]$. However, in this case, our dynamic programming table stores~$\DP_x[P,f,c,s] = \texttt{true}$ according to the definition of the computation.
	\end{claimproof}
	
    When computing the table for a leaf node~$x$, we just need to try all possible size functions~$s$, and for each of them, we check whether~$s = \mathbf{0}$ or not. Since the domain of~$s$ is of cardinality~$\parts$, we need at most~$\parts$ comparisons to determine the correct value. Therefore, the computation for a leaf node takes~$\Oh{n^\parts \cdot \parts}$ time.
	
	\paragraph{Introduce Node.} In introduce nodes, we try to put the introduced agent~$a$ in every possible part~$\pi_i$,~$i\in[\parts]$. For every possibility, we just verify whether this is fair for the agent~$a$. For the remaining agents, the verification can be easily done by checking the dynamic programming table of the child~$y$. This is implemented as follows.

    \begin{equation*}
        \DP_x[P,f,c,s] = \begin{cases}
            \texttt{false} & \text{if } \exists i\in[k]\colon c(a,i) \not= 0,\\
            \texttt{false} & \text{if } \sum_{i=1}^k f(a,i) = |F(a)\setminus\beta(x)|,\\
            \texttt{false} & \text{if } \operatorname{\phi-fair}(a, P, f) = \texttt{false}\text{, and}\\
            \DP_y[ \left.P\right\vert_{\beta(x)\setminus\{a\}}, f', \left.c\right\vert_{(\beta(x)\times[k]) \setminus (\{a\}\times[k])}, s ] & \text{otherwise,}
        \end{cases}
    \end{equation*}
    where~$f'$ is a function with~$\mathcal{D}(f') = \beta(x)\setminus\{a\}$ such that~$f'(b,i) = f(b,i)$ whenever~$b\not\in F(a)$ or~$i\not= P(a)$, and~$f'(b,i) = f(b,i) + 1$ otherwise.
	
	\begin{claim}\label{lem:tw:intro:correct}
		Let~$(P,f,c,s)$ be a signature and assume that the dynamic programming table~$\DP_y$ is computed correctly.
        Then~$\DP_x[P,f,c,s] = \texttt{true}$ if and only if there exists a partial~$\parts$-partition~$\pttn^x$ corresponding to the signature~$(P,f,c,s)$.
	\end{claim}
    \begin{claimproof}
        First, assume that~$\DP_x[P,f,c,s] = \texttt{true}$ and let~$i = P(a)$. Then, by the definition of the computation,~$c(a,\bullet) = 0$,~$\sum_{i=1}^{\parts} f(a,i) = |F(a)\setminus\beta(x)|$,~$\operatorname{\phi-fair}(a,P,f,c) = \texttt{true}$, and~$\DP_y[P',f',c',s] = \texttt{true}$, where~$P' = \left.P\right\vert_{\beta(x)\setminus\{a\}}$,~$f'(b,j) = f(b,j)$ whenever~$b\notin F(a)$ or~$j\not= i$ and~$f'(b,i) = f(b,i) + 1$ otherwise, and~$c' = \left.c\right\vert_{(\beta(x)\times[k]) \setminus (\{a\}\times[k])}$. Since we assume that~$\DP_y$ is computed correctly, for the signature~$(P',f',c',s)$, there exists a corresponding partial~$\parts$-partition~$\pttn^y$. Now, we create a partial~$\parts$-partition~$\pttn^x$ by setting~$\pttn^x_j = \pttn^y_j$ for every~$j\in[\parts]\setminus\{i\}$ and~$\pttn^x_i = \pttn^y_i \cup \{a\}$, and we claim that it corresponds to the signature~$(P,f,c,s)$.
        \begin{enumerate}
            \item For every~$b\in\beta(x)\setminus\{a\}$, we have~$\pttn^x(b) = \pttn^y(b) = P'(b) = P(b)$. For agent~$a$ the property trivially holds as we set~$\pttn^x(a) = P(a)$ during the construction of~$\pttn^x$.

            \item Let~$b\in\beta(x)$ and~$j\in[\parts]$. If~$b=a$, then~$|(\pttn^x_j\cap F(a))\setminus\beta(x)| = |\emptyset| = 0 = c(a,j)$, since there cannot be an already forgotten friend of~$a$ in any part by the definition of the tree decomposition. If~$b\not=a$, then the number of already forgotten friends of~$b$ in each part remains the same, so the property follows directly from the fact that~$c'(b,j) = c(b,j)$.

            \item Let~$b\in\beta(x)$ be a bag agent. If~$b=a$, condition~$|F(a)\setminus\beta(x)| = \sum_{j=1}^{\parts} c(a,j) + f(a,j)$ trivially holds, since~$c(a,\bullet) = 0$ and~$\sum_{j=1}^{\parts} f(a,j) = |F(a)\setminus\beta(x)|$ -- if any of these properties would be false, then the dynamic programming table stores \texttt{false} by the first or second condition in the definition of the computation. Hence, assume that~$b\not=a$ and~$b\not\in F(a)$. Since~$a$ is not a friend of~$b$, then~$f' = f$,~$c'=c$, and the sets of friends in each part remain the same. That is, the property is satisfied for such~$b$. Finally, let~$b\in F(a)$. Then we have~$|F(b)\setminus\beta(x)| = |F(b)\setminus(\beta(y)\cup\{a\})| = |F(b)\setminus\beta(y)| - 1$. We substitute~$\sum_{j=1} c'(b,j) + f'(b,j)$ for~$|F(b)\setminus\beta(y)|$ and obtain~$(\sum_{j=1} c'(b,j) + f'(b,j)) - 1$. However, we have that~$c'(b,\bullet) = c(b,\bullet)$,~$f'(b,j) = f(b,j)$ for every~$j\not=i$, and~$f'(b,i) = f(b,i) + 1$. We can therefore rewrite the equation to~$(\sum_{j=1} c(b,j) + f(b,j)) + 1 - 1$, which is~$\sum_{j=1} c(b,j) + f(b,j)$.

            \item The number of already forgotten agents is still the same in every part and so is~$s$ in both signatures. Property therefore follows from~$\pttn^y$ having exactly~$s(j)$ forgotten agents in every part~$j$.

            \item[5.+6.] In every fairness notion studied in this work, agents care only about the number of their friends in each part. Therefore, the fairness of no~$b\in V^x$ such that~$b\not\in F(a)$ is affected by the assignment of~$a$ to~$\pttn^x_i$, meaning that if~$\pttn^y$ is (conditionally)~$\phi$-fair, so is~$\pttn^x$ (recall that~$f'(b,\bullet) = f(b,\bullet)$ for every~$b\not\in F(a)$). For~$a$,~$\pttn^x$ is conditionally~$\phi$-fair, otherwise the third condition in the definition of the computation of~$\DP_x[P,f,c,s]$ would be valid and the stored value would be \texttt{false}, which is a contradiction. Therefore, the only case we need to consider is when~$b\in\beta(x)\cap F(a)$. Then,~$f'(b,j) = f(b,j)$ for every~$j\not= i$ and~$f(b,i) = f(b,i) - 1$. That is, we decreased the number of promised future agents in~$f$ by one. However, we extended~$\pttn^y_i$ by adding~$\{a\}$, so the overall number of~$b$'s friends in each part is still the same as in~$\pttn^y$ (considering~$f$).
        \end{enumerate}
        Overall,~$\pttn^x$ satisfies all the required properties, showing that it corresponds to the signature~$(P,f,c,s)$. That is, the left-to-right implication is correct.

        In the opposite direction, let a partial~$\parts$-partition~$\pttn^x$ correspond to~$(P,f,c,s)$ and let~$i = P(a)$. Since the agent~$a$ is introduced in node~$x$,~$V^x$ contains no already forgotten friend of~$a$. Therefore,~$c(a,\bullet)$ must be zero and it must also holds that~$\sum_{j=1}^{\parts} f(a,j) = |F(a)\setminus\beta(x)|$, otherwise~$\pttn^x$ breaks property 3. Finally,~$\pttn^x$ is~$\phi$-fair for~$a$ with respect to the promise~$f$ by property 6. By the definition of the computation in~$\DP_x$, the value of~$\DP_x[P,f,c,s]$ is computed as of case 4 of the recurrence. We now show that~$\DP_y[P',f',c',s]$, where~$P' = \left.P\right\vert_{\beta(x)\setminus\{a\}}$,~$f'(b,j) = f(b,j)$ whenever~$b\notin F(a)$ or~$j\not= i$ and~$f'(b,i) = f(b,i) + 1$ otherwise, and~$c' = \left.c\right\vert_{(\beta(x)\times[k]) \setminus (\{a\}\times[k])}$, is \texttt{true} and therefore, also~$\DP_x[P,f,c,s] = \texttt{true}$. We construct~$\pttn^y$ by removing~$a$ from~$\pttn^x$ and prove that~$\pttn^y$ corresponds to the signature~$(P',f',c',s)$ in~$y$. We again verify all the properties. First, the bag agents are assigned to the same parts in both~$\pttn^x$ and~$\pttn^y$ and~$P(b) = P'(b)$ for every~$b\in\beta(y)$, so the first property is clearly satisfied. As~$V^x\setminus\beta(x) = V^y\setminus\beta(y)$ we have that 
        \begin{enumerate}
            \item[2.] the number of forgotten friends for every bag agent~$b$ is the same in~$\pttn^x$ and~$\pttn^y$,~$c(b,\bullet) = c'(b,\bullet)$, so the second property follows,
            \item[3.] for every bag agent~$b\in\beta(y)$ and every~$j\in[\parts]$ we have~$|(\pttn^y_j\cap F(b))\setminus\beta(x)| = |(\pttn^x_j\cap F(b))\setminus\beta(x)| = c(b,j) = c'(b,j)$,
            \item[4.] the number of forgotten agents in each part remains the same and we use the same function~$s$ in both signatures, and
            \item[5.] no already forgotten agent is a friend of~$a$, meaning that the removal of~$a$ cannot make them envious.
        \end{enumerate}
        Finally, for the last property, it remains to observe that for every~$b\in\beta(y)$ which is non-friend of~$a$, the promise is the same and fairness is not affected by the presence of~$a$ in~$\pttn^x_i$ and~$\pttn^y_i$, respectively. For~$b\in F(a)$, the partition~$\pttn^x$ is~$\phi$-fair. In~$\pttn^y$, agent~$b$ lose one friend in part~$\pttn^y_i$. On the other hand, if we set~$f'(b,i) = f(b,i) + 1$, we compensate for this loss and~$\pttn^y$ is necessarily (conditionally)~$\phi$-fair for~$b$ with respect to~$f'$. That is,~$\pttn^y$ corresponds to the signature~$(P',f',c',s)$ in~$y$ and, since we assume that~$\DP_y$ is computed correctly,~$\DP_y[P',f',c',s]$ must be set to \texttt{true}. Consequently, also~$\DP_x[P,f,c,s]$ is set to \texttt{true}, which finishes the proof.
    \end{claimproof}

    For the running time, in every cell of~$\DP_x$, we verify three conditions and question at most one cell of~$\DP_y$. This can be done in linear time. As there are~$n^\Oh{\parts\cdot\tw}$ different cells, the computation of~$\DP_x$ for an introduce node~$x$ takes~$n^\Oh{\parts\cdot\tw}$ time.
	
	\paragraph{Forget Node.} A node~$x$ forgetting an agent~$a\not\in\beta(x)$ has exactly one child~$y$ such that~$\beta(y) = \beta(x) \cup \{a\}$. Intuitively, we check whether there exists at least one part~$i\in[k]$ such that if we allocate~$a$ to~$\pttn_i$, the partition is fair and all neighbors of~$a$ are present in~$V^x$. Such information can be easily read from the dynamic programming table of the child~$y$. Formally, the computation is defined as follows.
	
	\begin{multline*}
	    \DP_x[ P, f, c, s ] = 
            \bigvee_{i\in[k]}
            \bigvee_{\substack{
                (c_1,\ldots,c_{\parts})\\
                \sum_{j=1}^{\parts} c_j = |F(a)\setminus\beta(x)|
            }}\\
            \DP_y[ P \Join (v \mapsto i), f \Join ((a,\bullet) \mapsto 0), c \Join ((a,j)\mapsto c_j) \Join ( F(a)\times \{i\} \mapsto (c(\bullet,i) - 1) ),\\ s \Join ( i \mapsto ( s(i) - 1 ) ) ],
	\end{multline*}
    where, for a function~$g$, we use~$g \Join (x \mapsto y)$ to denote a function~$g'$ which is almost identical to~$g$ with the only difference being that for the value~$x$,~$g'$ returns the image~$y$.
	
	\begin{claim}\label{lem:tw:forget:correct}
		Let~$(P,f,c,s)$ be a signature and assume that the dynamic programming table~$\DP_y$ is computed correctly.
        Then~$\DP_x[P,f,c,s] = \texttt{true}$ if and only if there exists a partial~$\parts$-partition~$\pttn^x$ corresponding to the signature~$(P,f,c,s)$.
	\end{claim}
    \begin{claimproof}
        First, assume that~$\DP[P,f,c,s] = \texttt{true}$. Then, by the definition of the computation, there exist~$i\in[k]$ and~$(c_1,\ldots,c_\parts)$ such that~$\sum_{j=1}^\parts c_j = |F(a)\setminus\beta(x)|$ such that~$\DP_y[P',f',c',s'] = \texttt{true}$, where~$P' = P\Join(v\mapsto i)$,~$f' = f \Join ((a,\bullet)\mapsto 0)$,~$c' = c \Join ((a,j) \mapsto c_j) \Join (F(a)\times\{i\} \mapsto (c(\bullet,i)-1))$, and~$s' = s \Join (i \mapsto (s(i) - 1))$. Since we assumed that~$\DP_y$ is computed correctly, there exists a partial~$\parts$-partition~$\pttn^y$ compatible with~$(P',f',c',s')$ in~$y$. We set~$\pttn^x = \pttn^y$ and claim that~$\pttn^x$ is a partial~$\parts$-partition corresponding to~$(P,f,c,s)$. We need to verify all conditions from the definition of partial~$\parts$-partition.
        \begin{enumerate}
            \item Let~$b\in\beta(x)$ be a bag agent. The function~$P'$ returns for every~$a\in\beta(x)$ the value~$P(a)$ and since we take~$\pttn^x = \pttn^y$, it holds that~$\pttn^x(b) = \pttn^y(b) = P'(b) = P(b)$.
            
            \item Let~$b\in\beta(x)$ be a bag agent. If~$b\not\in F(a)$, then trivially~$|(\pttn^x_j\cap F(b))\setminus\beta(x)| = |(\pttn^y_i\cap F(b))\setminus\beta(y)\cup\{a\}|$ for every~$j\in[\parts]$. This can be simplified as~$|(\pttn^y_j\cap F(b))\setminus\beta(y)|$ since~$a\not\in F(b)$. This value is equal to~$c'(b,j)$ which is, by the definition of~$c'$, equal to~$c(b,j)$. The argumentation is the same for every~$b\in F(a)$ and~$j\in[\parts]\setminus\{i\}$. Finally, let~$b\in F(a)$ and~$j=i$. Then, we have~$|(\pttn^x_i \cap F(b))\setminus\beta(x)| = |((\pttn^y_i) \cap F(b))\setminus(\beta(y)\setminus\{a\})| = |((\pttn^y_i) \cap F(b))\setminus(\beta(y))| + |\{a\}| = c'(b,i) + 1 = c(b,i) - 1 + 1 = c(b,i)$, showing that the property is satisfied.
            
            \item Let~$b\in\beta(x)$ be a bag agent. If~$b\not\in F(a)$, then~$|F(b)\setminus\beta(x)| = |F(b)\setminus\beta(y)\setminus\{a\}| = |F(b)\setminus\beta(y)| = \sum_{j=1}^{\parts} c'(b,j) + f(b,j) = \sum_{j=1}^{\parts} c(b,j) + f(b,j)$. Therefore, the interesting case is when~$b\in F(a)$. Then, we have~$|F(b)\setminus\beta(x)| = |F(b)\setminus(\beta(y)\setminus\{a\})| = |F(b)\setminus\beta(y)| + |\{a\}| = (\sum_{j=1}^{\parts} c'(b,j) + f(b,j)) + 1$. For all~$j\in[\parts]\setminus\{i\}$, we have~$c'(b,j) = c(b,j)$ and we can split the sum into two parts~$\sum_{j\in[\parts]\setminus\{i\}}^{\parts} c'(b,j) + f(b,j)$ and~$c'(b,i) + f(b,i)$. Overall, we obtain~$(\sum_{j\in[\parts]\setminus\{i\}}^{\parts} c'(b,j) + f(b,j)) + c'(b,i) + f(b,i) + 1$ and by substituing~$c'(b,i) = c(b,i) - 1$, we obtain~$(\sum_{j\in[\parts]\setminus\{i\}}^{\parts} c'(b,j) + f(b,j)) + c(b,i) - 1 + f(b,i) + 1 = (\sum_{j\in[\parts]\setminus\{i\}}^{\parts} c'(b,j) + f(b,j)) + c(b,i) + f(b,i) = \sum_{j=1}^{\parts} c(b,j) + f(b,j)$, as desired.

            \item Let~$j\in[\parts]$ be a part. If~$j\not= i$, then~$s'(j)=s(j)$ and there is no new already forgotten agent in~$\pttn_j$. Therefore, let~$j = i$. Then agent~$a$ is newly forgotten and we have~$|\pttn^x_i\setminus\beta(x)| = |\pttn^y_i\setminus(\beta(y)\setminus\{a\})| = |\pttn^y_i\setminus\beta(y)| + |\{a\}| = s'(i) + 1$ and since~$s'(i) = s(i) - 1$, by substituting we obtain that~$|\pttn^x_i\setminus\beta(x)| = s(i) - 1 + 1 = s(i)$.

            \item Let~$b\in V^x\setminus\beta(x)$ be a forgotten agent. For every~$b\not= a$, the partition~$\pttn^y$ is~$\phi$-fair and so is~$\pttn^x$. For~$a$, the partition~$\pttn^y$ is conditionally~$\phi$-fair assuming that~$f'(a,j)$ friends of~$a$ are added to each part~$j\in[\parts]$. However,~$f'(a,\bullet) = 0$, so both~$\pttn^y$ and~$\pttn^x$ are unconditionally~$\phi$-fair for~$a$.

            \item For every~$b\in\beta(x)$ the conditional~$\phi$-fairness of~$\pttn^x$ trivially holds since~$f(b,\bullet) = f'(b,\bullet)$,~$\pttn^x = \pttn^y$, and~$\pttn^y$ is conditionally~$\phi$-fair for every~$b\in\beta(y)$.
        \end{enumerate}
        Consequently,~$\pttn^x$ is indeed a partial~$\parts$-partition corresponding to the signature~$(P,f,c,s)$, and the left to right implication therefore holds.

        In the opposite direction, let~$\pttn^x$ be a partial~$\parts$-partition corresponding to a signature~$(P,f,c,s)$ and~$i = \pttn^x(a)$. We construct~$\pttn^y = \pttn^x$ and claim that~$\pttn^y$ is a partial~$\parts$-partition corresponding to the signature~$(P',f',c',s')$, where~$P' = P \Join (a \mapsto i)$,~$f' = f \Join ((a,\bullet) \mapsto 0)$,~$c' = c \Join ( (a,j) \mapsto (\pttn^x_j\cap F(a)\setminus\beta(x) ) \Join (F(a)\times\{i\} \mapsto c(\bullet,i)-1)$, and~$s' = s \Join (i \mapsto (s(i) - 1))$, in~$y$. We again verify all the properties needed for~$\pttn^y$ to be partial~$\parts$-partition corresponding to~$(P',f',c',s')$ in~$y$.
        \begin{enumerate}
            \item Let~$b\in\beta(y)$ be a bag agent. If~$b = a$, then trivially~$\pttn^y(a) = i = P'(a)$ by the definition of~$P'$. If~$b\not= a$, then~$P'(b) = P(b) = \pttn^x(b) = \pttn^y(b)$.

            \item Let~$b\in\beta(y)$ be a bag agent and~$j\in[\parts]$ be a part. First, let~$b = a$. Then, we have~$|\pttn^y_j \cap F(a) \setminus\beta(y)| = |\pttn^x_j \cap F(b) \setminus (\beta(x)\cup\{a\})| = |\pttn^x_j \cap F(b) \setminus \beta(x)| = c'(a,j)$ by the definition of~$c'$ for~$a$. Next, let~$b\not= a$. If~$b\not\in F(a)$ or~$j\not=i$, then~$c'(b,\bullet) = c(b,\bullet)$ and either~$a$ is not a friend of~$b$ or~$a$ not in~$\pttn^y_j$, so the second condition trivially holds. Finally, let~$b\in F(a)$ and~$j=i$. Then~$|\pttn^y_i \cap F(b) \setminus\beta(y)| = |\pttn^x_i \cap F(b) \setminus (\beta(x)\cup\{a\})| = |\pttn^x_i \cap F(b) \setminus \beta(x)| - 1 = c(b,i) - 1$ and by the definition of~$c(b,i) = c'(b,i) + 1$, we have~$|\pttn^y_i \cap F(b) \setminus\beta(y)| = c'(b,i) + 1 - 1 = c'(b,i)$.

            \item Let~$b\in\beta(y)$ be a bag agent. If~$b=a$, then~$|F(a)\setminus\beta(y)| = |F(a)\setminus\beta(x)| = \sum_{j=1}^{\parts} |\pttn^x_j\cap F(a)\setminus\beta(x)| = \sum_{j=1}^{\parts} c'(a,j) = \sum_{j=1}^{\parts} c'(a,j) + 0 = \sum_{j=1}^{\parts} c'(a,j) + f'(a,j)$. For~$b\not=a\not\in F(a)$, we have~$|F(b)\setminus\beta(y)| = |F(b)\setminus\beta(x)| = \sum_{j=1}^{\parts} c(b,j) + f(b,j)$ and since in this case,~$c'(b,j) = c(b,j)$ and~$f'(b,j) = f(b,j)$, we clearly obtain~$|F(b)\setminus\beta(y)| = \sum_{j=1}^{\parts} c'(b,j) + f'(b,j)$. Finally, if~$b\in F(a)$, then~$|F(b)\setminus\beta(y)| = |F(b)\setminus(\beta(x)\cup\{a\})| = |F(b)\setminus\beta(x)| - 1 = (\sum_{j=1}^{\parts} c(b,j) + f(b,j)) - 1$. However,~$f'(b,j) = f(b,j)$,~$c(b,j) = c'(b,j)$ for every~$j\not=i$, and~$c(b,j) = c'(b,j) + 1$ and by substituting for~$c$ and~$f$, and taking~$+1$ outside of the sum, we obtain that the equation is equal to~$\sum_{j=1}^{\parts} c'(b,j) + f'(b,j)$.

            \item Let~$j\in[\parts]$ be a part. It must hold that~$|\pttn^y_j\setminus\beta(y)| = s'(j)$. First, let~$i=j$. Then,~$|\pttn^y_j\setminus\beta(y)| = |\pttn^x_j\setminus(\beta(x)\cup\{a\})| = |\pttn^x_j\setminus\beta(x)| - 1 = s(j) - 1 = s'(j)$. If~$j\not=i$, then~$a$ is never a member of~$\pttn^y_j$, so the equality follows from the fact that~$s'(j) = s(j)$.

            \item[5.+6.] The partition~$\pttn^y = \pttn^x$ is (conditionally)~$\phi$-fair for every agent in~$V^x\setminus\{a\}$ and we keep all promises valid. For~$a$, the partition~$\pttn^x$ is unconditionally~$\phi$-fair. Since we use~$f'(a,\bullet) = 0$, the partition~$\pttn^y$ is necessarily (conditionally)~$\phi$-fair for~$a$ with respect to~$f'$.
        \end{enumerate}
        Since, by our assumptions,~$\DP_y$ is computed correctly,~$\DP_y[P',f',c',s']$ must be set to~$\texttt{true}$. Therefore, also~$\DP_x[P,f,c,s]$ necessarily stores \texttt{true}, as the cell~$\DP_y[P',f',c',s']$ is considered when computing the value of~$\DP_x[P,f,c,s]$ by the definition of the computation.
    \end{claimproof}

    In every cell, for every~$i\in[\parts]$, we try all vectors~$(c_1,\ldots,c_\parts)$. This gives us~$\parts\cdot n^\Oh{\parts} = n^\Oh{\parts}$ possibilities. Each possibility can be checked in constant time by questioning~$y$'s table. As there are~$n^\Oh{\parts\cdot\tw}$ different cells, the overall computation of the dynamic programming table~$\DP_x$ for a forget node~$x$ takes~$n^\Oh{\parts\cdot\tw}$ time.
	
	\paragraph{Join Node.} Join node~$x$ has exactly two children~$y$ and~$z$ and it holds that~$\beta(x) = \beta(y) = \beta(z)$. Moreover, as each bag forms a separator in~$G$, we have that~$V^y \cap V^z = \beta(x)$, that is, no agent forgotten in~$V^y$ appears in~$V^z$ and additionally, all paths from~$V^y$ to~$V^z$ necessarily traverse~$\beta(x)$. Therefore, for a signature~$(P,f,c,s)$, we just try all possible distributions of already forgotten vertices between the sub-trees rooted in~$y$ and~$z$. Formally, the computation is defined as follows.  
	\begin{equation*}
		\DP_x[ P, f, c, s ] = 
        \bigvee_{\substack{
            f_y,f_z,c_y,c_z\colon \beta(x)\times[k]\to[n]\\
            c_y + c_z = c\\
            f_z - c_y = f_y - c_z = f
        }}
        \bigvee_{\substack{
            s_y,s_z\colon [k]\to[n]\\
            s_y + s_z = s
        }}
			\left(
				\DP_y[ P, f_y, c_y, s_y ] \land
				\DP_z[ P, f_z, c_z, s_z ]
			\right)
	\end{equation*}
	
	\begin{claim}\label{lem:tw:join:correct}
        Let~$(P,f,c,s)$ be a signature and assume that the dynamic programming tables~$\DP_y$ and~$\DP_z$ are computed correctly.
        Then~$\DP_x[P,f,c,s] = \texttt{true}$ if and only if there exists a partial~$\parts$-partition~$\pttn^x$ corresponding to the signature~$(P,f,c,s)$.
	\end{claim}
	\begin{claimproof}
        First, suppose that, for the signature~$(P,f,c,s)$, the dynamic programming table~$\DP_x$ stores \texttt{true}. Then, by the definition of computation, there exist~$f_y$,~$f_z$,~$c_y$,~$c_z$,~$s_y$, and~$s_z$ such that for every~$a\in\beta(x)$ and~$i\in[\parts]$ we have~$c_y(a,i) + c_z(a,i) = c(a,i)$ and~$f_z(a,i) - c_y(a,i) = f_y(a,i) - c_z(a,i) = f(a,i)$,~$s_y(i) + s_z(i) = s(i)$, and~$\DP_y[P,f_y,c_y,s_y] = \DP_z[P,f_z,c_z,s_z] = \texttt{true}$. By our assumptions that~$\DP_y$ and~$\DP_y$ are computed correctly, there must exist~$\pttn^y$ and~$\pttn^z$ corresponding to the signature~$(P,f_y,c_y,s_y)$ in~$y$ and~$(P,f_z,c_z,s_z)$ in~$z$, respectively. We define a partial~$\parts$-partition~$\pttn^x$ such that we set~$\pttn_i^x = \pttn^y_i \cup \pttn^z_i$ for every~$i\in[\parts]$ and claim that~$\pttn^x$ corresponds to the signature~$(P,f,c,s)$. Next, we verify that~$\pttn^x$ indeed satisfies all the required properties.

        \begin{enumerate}
            \item Since the partition of the bag agents in all three partial~$\parts$-partitions is the same, each bag agent is assigned to exactly one part~$P(a)$.

            \item Let~$a\in\beta(x)$ be a bag agent and~$i\in[\parts]$ be a part. By construction, we have~$|(\pttn^x_i\cap F(a))\setminus\beta(x)| = |((\pttn^y_i \cup \pttn^z_i) \cap F(a))\setminus \beta(x)| = |((\pttn^y_i \cap F(a)\setminus \beta(y))| + |((\pttn^z_i \cap F(a)\setminus \beta(z))|$, where adjustments are possible due to the facts that~$V^y\setminus\beta(y) \cap V^y\setminus\beta(z) = \emptyset$ and~$\beta(x) = \beta(y) = \beta(z)$. The last expression can be substituted to~$c_y(a,i) + c_z(a,i)$, which is exactly~$c(a,i)$ by our conditions.

            \item Let~$a\in\beta(x)$ be a bag agent. We want to check that~$\sum_{i=1}^k c(a,i) + f(a,i) = |F(a)\setminus\beta(x)|$. The left part of the equation can be rewritten as~$\sum_{i=1}^\parts c_y(a,i) + c_z(a,i) + f_y(a,i) - c_z(a,i)$ which can be simplified to~$\sum_{i=1}^\parts c_y(a,i) + f_y(a,i)$. Since~$\pttn^y$ corresponds to the signature~$(P,f_y,c_y,s_y)$, this is exactly~$|F(a)\setminus\beta(y)| = |F(a)\setminus\beta(x)|$.

            \item Let~$i\in[\parts]$ be a part. Then we have~$|\pttn^x_i\setminus\beta(x)| = |(\pttn^y_i \cup \pttn^z_i)\setminus\beta(x)| = |\pttn^y_i\setminus\beta(y)| + |\pttn^z_i\setminus\beta(z)|$, and by the correctness of~$\pttn^y$ and~$\pttn^z$, we can substitute to obtain~$s_y(i) + s_z(i)$. However, this is exactly equal to~$s(i)$.

            \item All friends of each already forgotten agent~$a\in V^x\setminus\beta(x)$ are assigned the same part as in~$\pttn^y$ or~$\pttn^z$, so the fairness for these agents remains satisfied.

            \item Let~$a\in\beta(x)$ be a bag agent. Since~$\pttn^y$ is a partial~$\parts$-partition corresponding to~$(P,f_y,c_y,s_y)$, agent~$a$ must consider~$\pttn^y$ fair if each part~$i\in[\parts]$ contains~$f_y(a,i)$ additional friends of~$a$. The value of~$f(a,i)$ is equal to~$f_y(a,i) - c_z(a,i)$, which means that the promise of future agents added to~$i$ is lowered; however, each part~$i$ contains exactly~$c_z(a,i)$ new already forgotten friends of~$a$ who compensate for the loss in~$f(a,i)$. Consequently, the overall number of~$a$'s friends in every part is the same and therefore~$\pttn^x$ is, assuming~$f$, also necessarily~$\phi$-fair.
        \end{enumerate}
        As we verified all the properties, we see that~$\pttn^x$ is indeed a partial~$\parts$-partition corresponding to the signature~$(P,f,c,s)$. So, the left-to-right implication is correct.
        
        In the opposite direction, let there exist a partial partition~$\pttn^x$ corresponding to the signature~$(P,f,c,s)$. We define two partial~$\parts$-partitions~$\pttn^y$ and~$\pttn^z$ such that we set~$\pttn^y_i = \pttn^x_i \cap V^y$ and~$\pttn^z_i = \pttn^x_i \cap V^z$, respectively. We claim that~$\pttn^y$ corresponds to a signature~$(P,f_y,c_y,s_y)$ in~$y$ and~$\pttn^z$ corresponds to a signature~$(P,f_z,c_z,s_z)$ in~$z$, respectively, such that for all~$a\in\beta(x)$ and~$i\in[\parts]$ it holds that~$f_y(a,i) = f(a,i) + |(\pttn^x_i \cap F(a))\setminus V^y|$,~$c_y(a,i) = |(\pttn^y_i\cap F(a))\setminus\beta(y)|$,~$s_y(i) = |(\pttn^x_i \cap V^y)\setminus\beta(x)|$, and similarly for~$f_z$,~$c_z$, and~$s_z$. If these indeed hold, then, for the signature~$(P,f,c,s)$, our dynamic programming table necessarily stores \texttt{true} by the definition of the computation, which finishes the proof. We verify the properties for~$\pttn^y$ and~$(P,f_y,c_y,s_y)$, as for~$\pttn^z$ and~$(P,f_z,c_z,s_z)$, the properties follow by symmetric arguments. First, the correctness of the partitioning of bag agents is trivial. Let~$a\in\beta(y)$ be a bag agent and~$i\in[\parts]$ be a part. We have~$|(\pttn^y_i \cap F(a))\setminus\beta(y)|$, which is exactly equal to~$c_y(a,i)$ by the definition of~$c_y$. Next,~$\sum_{i=1}^{\parts} c_y(a,i) + f_y(a,i) = \sum_{i=1}^{\parts} |(\pttn^y_i \cap F(a))\setminus\beta(y)| + f(a,i) + |(\pttn^x_i \cap F(a))\setminus V^y| = \sum_{i=1}^{\parts} |(\pttn^y_i \cap F(a))\setminus\beta(y)| + |(\pttn^x_i \cap F(a))\setminus V^y| + f(a,i)= \sum_{i=1}^{\parts} |(\pttn^x_i \cap V^y \cap F(a))\setminus\beta(y)| + |(\pttn^x_i \cap F(a))\setminus V^y| + f(a,i)$. As both sets are disjoint, we can merge them into one counting and we obtain~$\sum_{i=1}^{\parts} |(\pttn^x_i \cap F(a))\setminus \beta(x)| + f(a,i)$. However,~$|(\pttn^x_i \cap F(a))\setminus \beta(x)|$ is exactly~$c(a,i)$, so we get~$\sum_{i=1}^{\parts} |(\pttn^x_i \cap F(a))\setminus \beta(x)| + f(a,i) = \sum_{i=1}^{\parts} c(a,i) + f(a,i) = |F(a)\setminus\beta(x)|$. Let~$i\in[\parts]$ be a part. The number of already forgotten vertices in~$\pttn^y_i$ is~$|\pttn_i^y\setminus\beta(y)| = |(\pttn^x_i\cap V^y) \setminus \beta(x)|$, which is exactly how we defined~$s_y(i)$. The partition~$\pttn^y$ is also clearly~$\phi$-fair for every already forgotten agent~$a\in V^y\setminus\beta(y)$, as none of their friends is removed from its original part in~$\pttn^x$. For bag vertices, the loss in numbers of friends in each part is compensated in the promise~$f$, so it again follows from the fact that~$\pttn^x$ is conditionally~$\phi$-fair with respect to~$f$. It remains to show that our algorithm necessarily checked this combination of~$(P,f_y,c_y,s_y)$ and~$(P,f_z,c_z,s_z)$. First, we verify that~$c_y + c_z = c$. Let~$a\in\beta(x)$ be an agent and~$i\in[\parts]$ be a part. If we substitute for~$c_y$ and~$c_z$, we obtain~$|(\pttn^y_i\cap F(a))\setminus\beta(y)| + |(\pttn^z_i\cap F(a))\setminus\beta(z)| = |((\pttn^x_i\cap V^y)\cap F(a))\setminus\beta(x)| + |((\pttn^x_i\cap V^z)\cap F(a))\setminus\beta(x)|$. By the definition of the tree decomposition,~$V^y \cap V^z = \beta(x)$, so we can rewrite the expression as~$|(((\pttn^x_i \cap V^y) \cup (\pttn^x_i \cap V^z))\cap F(a))\setminus\beta(x)| = |(\pttn^x_i\cap F(a))\setminus\beta(x)|$, which is~$c(a,i)$ by the assumption that~$\pttn^x$ is valid partial~$\parts$-partition corresponding to~$(P,f,c,s)$. Next, we check that~$f_y(a,i) - c_z(a,i) = f_z(a,i) - c_y(a,i) = f(a,i)$ for every~$a\in\beta(x)$ and~$i\in[\parts]$. First, we substitute for~$f_y$ and~$c_z$ and obtain~$f(a,i) + |(\pttn^x_i \cap F(a))\setminus V^y| - |(\pttn^y_i\cap F(a))\setminus\beta(y)| = f(a,i) + |(\pttn^x_i \cap F(a))\setminus V^y| - |(\pttn^z_i\cap F(a))\setminus\beta(z)| = f(a,i) + |(\pttn^x_i\setminus V^y) \cap (F(a)\setminus V^y)| - |(\pttn^z_i\setminus\beta(x))\cap(F(a)\setminus\beta(x))|$. We replace~$\pttn^x_i\setminus V^y$ with~$\pttn^z_i\setminus\beta(x)$ to get~$f(a,i) + |(\pttn^z_i\setminus \beta(x)) \cap (F(a)\setminus V^y)| - |(\pttn^z_i\setminus\beta(x))\cap(F(a)\setminus\beta(x))|$. Now we can notice that~$V^y$ plays no role in the first expression, so it is equivalent to~$f(a,i) + |(\pttn^z_i\setminus \beta(x)) \cap F(a)| - |(\pttn^z_i\setminus \beta(x)) \cap F(a)| = f(a,i)$. So, the property holds and for~$f_z$ and~$c_y$, it is enough to follow the symmetric argument. Finally, we verify that~$s_y(i) + s_z(i) = s(i)$. We can substitute for~$s_y$ and~$s_z$ and get~$|(\pttn^x_i \cap V^y)\setminus\beta(x)| + |(\pttn^x_i \cap V^z)\setminus\beta(x)|$. This is clearly equal to~$|\pttn^x_i\setminus\beta(x)|$ as~$V^y\cap V^z = \emptyset$, and this is exactly~$s(i)$. That is, our algorithm checks this combination of signatures and, as both~$\DP_y[P,f_y,c_y,s_y]$ and~$\DP_x[P,f_z,c_z,s_z]$ are set to \texttt{true}, also sets~$\DP_x[P,f,c,s]$ to \texttt{true}.
    \end{claimproof}
	
	In every cell, we try all functions~$c_y$ (the choice of~$c_y$ completely determines~$c_z$),~$f_y$, and~$s_y$. This gives us~$n^\Oh{\parts\cdot\tw}\cdot n^\Oh{\parts\cdot\tw}\cdot n^k = n^\Oh{\parts\cdot\tw}$ possibilities. Each possibility can be checked in constant time by questioning the children's tables. As there are~$n^\Oh{\parts\cdot\tw}$ different cells, the overall computation of the dynamic programming table~$\DP_x$ for a join node~$x$ takes time~$n^\Oh{\parts\cdot\tw}$.
	
	\medskip
	
	The overall correctness of the algorithm follows by the induction over the types of the tree decomposition, as formalized in \Cref{lem:tw:leaf:correct,lem:tw:intro:correct,lem:tw:forget:correct,lem:tw:join:correct}. We have also shown that the computation in each node takes~$n^\Oh{\tw\cdot\parts}$ time, which concludes the proof.
\end{proof}

Since the family of trees is the only class of connected graphs of tree-width one, we directly obtain the desired result.

\begin{corollary}\label{thm:tree:XP}
     If~$G$ is a tree and the utilities are binary, then for every fairness notion~$\phi\in\{\text{EF},\text{EFX}_0,\text{EFX},\text{EF1},\text{PROP},\text{MMS}\}$, there is an~\XP algorithm parameterized by~$k$ that decides whether a partition which is fair with respect to~$\phi$ exists.
\end{corollary}

It is not hard to see that our above algorithms almost directly work for any symmetric and objective utilities -- not necessarily binary. In fact, with a little bit more of effort, the algorithm can be generalized to settings where the co-domain of every utility function is of constant size.

None of the previous results rule out better algorithms for weaker fairness notions such as EFX, EF1, or MMS. The next set of results shows strong existence guarantees and tractable algorithms that can find a fair outcome for these fairness notions, even if we assume the most general types of valuations; namely, we show that fair outcomes are guaranteed to exist for arbitrary monotone utilities.

We give an algorithm that outputs~$k$-partition that is MMS and EFX (end hence EF1) for any monotone utility function. In addition, if the utilities are additive, this algorithm can be implemented in linear time. 

The main idea of the algorithm is to root each tree in the forest in some arbitrary vertex and then process the agents
in a BFS order. When processing an agent~$a$, we are computing a ``preliminary'' assignment of all the children of the agent and finalizing the assignment of~$a$ to the parts of~$\pi$. 
Loosely speaking, we first determine the list~$\vec{\ell} = (\ell_1,\ldots, \ell_k)$, where~$\ell_i$,~$i\in [k]$, determines how many of~$|\text{children}(a)|$ many newly assigned agents should be assigned to~$\pi_i$ in order to keep the parts always balanced.
Note that, since we are keeping ``being balanced'' as an invariant, then~$|\ell_i-\ell_j|\le 1$ for all~$i,j\in [k]$. 
Afterwards,~$a$ decides on (1) whether to stay in the same part of~$\pi$ or change to~$\pi(p_a)$, where~$p_a$ is the parent of~$a$ (if this change is allowed by~$\vec{\ell}$), 
and (2) how to distribute the children among the parts according to~$\vec{\ell}$. 
If~$a$ was originally in~$\pi_j$ and chooses to change to~$\pi_i = \pi(p_a)$, then to preserve balancedness, we have to add~$\ell_i-1$ children of~$a$ to~$\pi_i$ and ~$\ell_j+1$ children of~$a$ to~$\pi_j$.
Since (a) the number of neighbors of~$a$ in each part is balanced, (b) agent~$a$ is choosing which children are in the same part with it, and (c) agent~$a$ is allowed to change to the part of its parent, it is not too difficult to show that this partial assignment is MMS and EFX for agent~$a$. 
From this point onward fairness for agent~$a$ is guaranteed: any change between parts for the children of~$a$ can only increase the utility of~$a$, hence MMS and EFX for~$a$ are preserved.

We formalize this more in the proof of the following theorem.

\begin{theorem}\label{thm:monotone:MMS:exists}
	If~$G$ is a forest and~$k\in \mathbb{N}$, there always exist a~$k$-partition that is MMS and EFX, even if the utility functions are monotone. Moreover, if the utilities are additive, we can find such a partition in linear time.
\end{theorem}
\begin{proof}
    We follow the main idea of the proof described above. Let~$n = |V(G)|$ and let~$\vec{b} = (b_1, b_2, \ldots, b_k)$ be an initial ``budget'' vector such that (1)~$b_i\in \{\lfloor \frac{n}{k}\rfloor, \lceil \frac{n}{k}\rceil\}$, for all~$i\in [k]$, and (2)~$\sum_{i\in [k]}b_k = n$. We will use the budget vector~$\vec{b}$ to keep track of how many of the unassigned vertices need to be assigned to each of the parts. Since part are symmetric, we can chose arbitrary vector~$\vec{b}$ that satisfies (1) and (2) as the initial budget vector. Throughout the execution of the algorithm, we will be assigning agents to the parts of partition~$\pi$ iteratively and we will preserve the invariant that both parts of (partially assigned) partition
    and the remaining budget kept in the budget vector are balanced. 
    
    Now in every tree~$T$ of~$G$, we choose arbitrary vertex~$r_T$ as the root of~$T$. We will now process~$G$ in a BFS order. 
    When we process a root~$r_T$, then: 
    \begin{enumerate}
        \item We assign~$r_T$ to a part~$i$ such that~$b_i=\max_{j\in [k]}b_j$.
        \item We compute an assignment vector~$\vec{\ell}= (\ell_1, \ldots, \ell_k)$ as follows: 
        \begin{enumerate}
            \item We initiate~$\ell_i=0$ for all~$i\in [k]$.
            \item In a for cycle with a range of the number of children of~$r_T$, we pick~$i$ such that~$b_i=\max_{j\in [k]}b_j$, then we increase~$\ell_i$ and decrease~$b_i$ by one.
        \end{enumerate}
        \item For~$i\in [k]$ such that~$r_T$ is assigned to~$\pi_i$, we find a subset~$S$ of the children of~$r_T$ such that~$|S| = \ell_i$ that maximizes~$\util_{r_T}(S)$.
        \item We assign agents in~$S$ to~$\pi_i$.
        \item For all~$j\in [k]$,~$j\neq i$, we assign arbitrary~$\ell_j$ many children that have not yet been assigned to~$\pi_j$.
    \end{enumerate}
    Let us first observe, that we always distribute the children to the parts in a balanced way. 
    \begin{claim}
        Let~$\vec{b}^1$ be the budget vector before we processed~$r_T$ and~$\vec{b}^2$ the budget vector after processing~$r_T$. If for all~$i,j\in [k]$ it holds~$|b^1_i - b^1_j|\le 1$, then for all~$i,j\in [k]$ it holds~$|b^2_i - b^2_j|\le 1$ and~$|\ell_i - \ell_j|\le 1$.
    \end{claim}
    \begin{claimproof}
        This straightforwardly follows from the procedure how~$\vec{\ell}$ is computed, always picking the largest part~$i$ in the budget and moving one from that part in budget vector~$\vec{b}$ to the assignment vector~$\vec{\ell}$. This makes~$i$ a smallest part that is now only assigned when each one of the largest parts is assigned again.  
    \end{claimproof}
    
    With the above claim, it is rather straightforward to verify that since we chose the subset of size~$\ell_i$ that maximizes the utility of the agent~$r_T$ and any other part gets at most~$\ell_i+1$ neighbors of~$r_T$, we get that the assignment satisfies the condition for EFX. We not that while the assignment to the children is not final yet, only operation that is allowed on children of~$r_T$ is to move it from some different part to~$\pttn(r_T)$, which only increases the utility of~$r_T$ for its part and decreases its utility of other parts. 
    \begin{claim}\label{clm:trees_EFX_root}
    After we assign~$r_T$ and its children as above, it holds for all~$b\in \agents$ and every~$c\in (\pttn(b)\setminus b)\cap \Fr(r_T)$ that
   ~$\util_{r_T}(\pttn(r_T)) \geq \util_{r_T}(\pttn(b)\setminus\{b,c\})$.
    \end{claim}
    \begin{claimproof}
        Let~$\pttn_i = \pttn(r_T)$ and~$\pttn_j = \pttn(b)$. From the definition of~$\vec{\ell}$, it follows that~$|\ell_i-\ell_j| \le 1$. Since~$c\in \pttn_j\cap \Fr(r_T)$, it follows that~$|(\pttn_j\setminus \{b,c\})|\le \ell_j-1 \le \ell_i$. The claim then follows from the fact that~$\util_{r_T}$ is monotone, that~$\util_{r_T}(\pttn(r_T))\ge \util_{r_T}(\pttn(S))$, and the selection of~$S$. 
    \end{claimproof}
    
    Similarly, we observe that in any partition of~$\Fr(r_T)$, at least one part has size at most~$\ell_i$, and we get that~$r_T$ is getting at least its MMS share. 
    \begin{claim}\label{clm:trees_MMS_share_root}
       After we assign~$r_T$ and its children as above, it holds that~$\util_{r_T}(\pttn(r_T)) \geq \operatorname{MMS-share}(a)$.
    \end{claim}
    \begin{claimproof}
        Let~$\pttn_i = \pttn(r_T)$ and observe that~$\ell_i \in \{\lfloor\frac{|\Fr(r_T)|}{k}\rfloor, \lceil\frac{|\Fr(r_T)|}{k}\rceil\}$. Now observe that every~$k$-partition~$\pi'$ of~$\Fr(r_T)$ contains a part of~$\pi'_j$ size at most~$\lceil\frac{|\Fr(r_T)|}{k}\rceil$. Since~$\util_{r_T}$ is monotone, it follows from the choice of~$S$ that~$\util_{r_T}(S) \ge \util_{r_T}(\pi'_j)$ and hence~$\util_{r_T}(\pttn(r_T)) = \util_{r_T}(S) \geq \operatorname{MMS-share}(r_T)$.
    \end{claimproof}
    
    Finally, we show that we can process~$r_T$ in linear time in the case of additive utilities. 
    \begin{claim}\label{clm:linear_time_root}
        If~$\util_{r_T}$ is additive, then we can process a root~$r_T$ in~$\Oh{|\Fr(r_T)|}$ time. 
    \end{claim}
    \begin{claimproof}
        To achieve linear time, we would start with the initial budget vector~$\vec{b}$ such that all values~$\lceil \frac{n}{k}\rceil$ are before ~$\lfloor \frac{n}{k}\rfloor$ and we would keep a pointer~$\operatorname{idx}$ that initially points to~$1$. This way, whenever we are looking for~$b_i$ with maximum value in step 1. or 2.(b), we can return~$b_{\operatorname{idx}}$ and increase~$idx$ by one modulo~$k$ (i.e.,~$idx := (idx \mod k) + 1$). It follows that we can perform the steps 1. and 2. in~$\Oh{|\Fr(r_T)|}$ time. To find~$S$ in the case of additive utilities, we only need to find~$\ell_i$ many largest valued agents in~$\Fr(r_T)$, which can be done in linear time~\cite{BlumFPRT73}. The distribution of agents in~$\Fr(r_T)$ to parts in step 5. can also be straightforwardly done in time~$\Oh{|\Fr(r_T)|}$.
    \end{claimproof}
    
    Now, when we process any other agent~$a\in \agents$, its parent~$p_a$ is already assigned to a part~$\pi_i = \pttn(p_a)$, the agent~$a$ is assigned to a part~$\pi_j = \pttn(a)$, and none of its children is assigned to any part.
    The processing of~$a$ is very similar to the processing of the root. 
    We again compute the vector~$\vec{\ell}$ and update the budget vector~$\vec{b}$ in the exactly the same manner as for a root in step 2.. Then we compute the set~$S$ as in step 3. However, if~$i\neq j$ and~$\ell_i \ge 1$, then we also compute a subset~$S'$ of the children of~$a$ such that~$|S'| = \ell_i - 1$ and the utility~$\util_{a}(S'\cup \{p_a\})$ is maximized. If~$\util_{a}(S) \ge \util_{a}(S'\cup \{p_a\})$, we proceed identically to the case of dealing with the root. Otherwise we replace steps 4. and 5. with the following 
    \begin{enumerate}
        \item[4'.] 
        \begin{enumerate}
            \item We move~$a$ from~$\pttn_j$ to~$\pttn_i$. 
            \item We assign agents in~$S'$ to~$\pttn_i$.
        \end{enumerate}
        \item[5'.]  
        \begin{enumerate}
            \item We assign arbitrary~$\ell_{i}+1$ many not yet assigned children to~$\pttn_{i'}$.
            \item For all~$i'\in ([k]\setminus\{i,j\})$ we assign arbitrary~$\ell_{i'}$ many not yet assigned children to~$\pttn_{i'}$
        \end{enumerate}
    \end{enumerate}
    
    Note that it the case of a non-root agent, after computation of~$\vec{\ell}$, it can happen that~$\ell_i = \ell_j+1$ and hence the number of agents in~$\Fr(a)$ that we would be assigned to~$\pi_i=\pttn(p_a)$ might be by two larger than the number of agents in~$\Fr(a)$ to be assigned to~$\pi_j = \pttn(a)$. However, in this case~$|S'|= |S|$ and hence ~$\util_{a}(S) \le \util_{a}(S'\cup \{p_a\})$ by the definition of~$S$ and~$S'$. We can now show the analogous claims for the utility of the agent~$a$ after we finish processing it. 
    
    \begin{claim}\label{clm:trees_EFX_ineer}
    After we assign~$a$ and its children as above, it holds for all~$b\in \agents$ and every~$c\in (\pttn(b)\setminus b)\cap \Fr(a)$ that
   ~$\util_{a}(\pttn(a)) \geq \util_{a}(\pttn(b)\setminus\{b,c\})$.
    \end{claim}
    \begin{claimproof}
    If we do not move~$a$ to the part of~$p_a$, then the proof here is basically analogous to the case when we processed a root. The only case we need to consider separately in this case is if~$\pttn(b) = \pttn(p_a)$. In this case we observe that~$|(\pttn(b)\setminus\{b,c,p_a\})\cap \Fr(a)|\le |S'|$ and by the choice of~$S$ and~$S'$ it follows~$\util_{a}(S)\ge \util_{a}(S') \ge \util_{a}(\pttn(b)\setminus\{b,c\})$ and the claim holds.
    
    On the other hand, when~$a$ moves from~$\pttn_j$ to~$\pttn_i = \pttn(p_a)$ is again relatively straightforward. Note that~$\util_{a}(S) < \util_{a}(S'\cup \{p_a\})$ and following the proof of the first case we get that if~$b$ is not in~$\pttn_i$ nor~$\pttn_j$, then~$\util_{a}(\pttn(a))\geq \util_{a}(S'\cup \{p_a\}) \geq \util_{a}(S) \geq \util_{a}(\pttn(b)\setminus\{b,c\})$.
    Finally, if~$b\in \pttn_j$, then~$|\pttn(b)\setminus\{b,c\}|\le \ell_j = |S|$ and by how~$S$ and~$S'$ were chosen, we again get~$\util_{a}(\pttn(a))\geq \util_{a}(S'\cup \{p_a\}) \geq \util_{a}(S) \geq \util_{a}(\pttn(b)\setminus\{b,c\})$.
    \end{claimproof}
    
    \begin{claim}\label{clm:trees_MMS_share_inner}
       After we assign~$a$ and its children as above, it holds that~$\util_{a}(\pttn(a)) \geq \operatorname{MMS-share}(a)$.
    \end{claim}
    \begin{claimproof}
    If number of neighbors of~$a$ in~$\pttn(a)$ is at least~$\lfloor\frac{|\Fr(a)|}{k}\rfloor$, then the proof follows analogously to the case when processing the root. Moreover, we can easily assume that this is the case. Only case in which this would not hold is when~$\ell_j = \frac{|\Fr(a)|}{k}-1$, and for all~$i'\in [k]$,~$i'\neq j$, we have~$\ell_{i'} = \frac{|\Fr(a)|}{k}$ (that is~$\frac{|\Fr(a)|}{k}+1$ many neighbors are in~$\pttn(p_a)$). However, as discussed above, in this case~$\util_{a}(S) \le \util_{a}(S'\cup \{p_a\})$ and either~$a$ actually moved to~$\pttn(p_a)$ and~$\frac{|\Fr(a)|}{k}$ many neighbors of~$a$ are in~$\pttn(a)$ or~$\util_{a}(S) = \util_{a}(S'\cup \{p_a\})$ and we can assume that this happened for the purpose of comparing with the MMS share of~$a$. 
    \end{claimproof}
    
    \begin{claim}\label{clm:linear_time_inner}
        If~$\util_{a}$ is additive, then we can process the vertex~$a$ in~$\Oh{|\Fr(a)|}$ time. 
    \end{claim}
    \begin{claimproof}
    The proof is analogous to the case when~$a$ is a root. Only additional computation here is to compute~$S'$, which is basically the same computation as that of~$S$ and can be done in~$\Oh{|\Fr(a)|}$ time. Besides that, we only need to compare the value of choosing~$S$ versus~$S'$ and assign the agents to the correct parts of the partition, which can also be easily done in~$\Oh{|\Fr(a)|}$ time.
    \end{claimproof}
    
    Note that the above discussion implicitly assumes that the vertex we are processing has children. However, leaves have only one neighbor, so their MMS share is zero, and removing the neighbor from a part results in a value of zero for that part. It follows that also leaves satisfy all the notions without further processing and always. Moreover, it is rather easy to see (from how we construct assignment vector~$\vec{\ell}$ for the children) that when we finish assigning all the vertices to parts, then the budget vector~$\vec{b}$ is the zero-vector and the resulting~$k$-partition is balanced. Hence, the theorem follows.
\end{proof}

\section{Beyond Balanced Sizes: Implications for Hedonic Games}\label{sec:discussion}

So far, we have investigated the setting where all the parts are almost the same size. One can argue that this requirement is too restrictive and fails to capture many real-life scenarios. Returning to our initial motivational example, it may be easily the case that not all tables are the same size, but we more likely have~$\parts$ tables with possibly variable capacities~$n_1,\ldots,n_\parts$. In this section, we discuss how our results generalize to the more general model where the desired sizes of all parts are given as part of the input.

The model described above with parts of prescribed sizes fits in the setting of \emph{additively separable hedonic games}~\cite{BogomolnaiaJ2002} (ASHGs for short). Generally, in ASHGs, there is an edge-weighted graph~$G$ over the set of agents. The utility some agent~$a$ has if it is assigned to a group of agents~$S$ is then computed simply as the sum of weights of all edges between~$a$ and other members of~$S$ and the goal here is to find a ``desirable'' partitioning of all agents into pair-wise disjoint groups. What does ``desirable'' mean in this context depends on the particular application and is formalized using some well-defined mathematical property; usually, we are interested in \emph{stable} partitions, where no agent (or group of agents) can deviate from the current partition and improve its utility, partitions that are optimizing some notion of economic \emph{efficiency}, such as the sum of all utilities, or combination of both. See, e.g., \cite{AzizBS2013} for a survey on solution concepts in ASHGs. Some papers also consider fairness to be the desirable solution concept in coalition formation~\cite{AzizBS2013,WrightV2015,Peters2016a,Peters2016b,Ueda2018,BarrotY2019,KerkmannNR2021}.

The most relevant modification of ASHGs for our work are the recently introduced \emph{additively separable hedonic games with fixed-size coalitions} (ASHGs-FSC for short), introduced recently by \citeauthor{BiloMM2022}~\cite{BiloMM2022}. These authors extended the basic ASHGs with a vector of size~$\parts$ that prescribe not only the number of coalitions we are allowed to create, but also the size of each coalition. Formally, they define the model as follows.

\begin{definition}
    An \emph{additively separable hedonic game with fixed-size coalitions} (ASHG-FSC for short) is a triple~$\Gamma=(H,\parts,\vec{n}=(n_i)_{i\in[\parts]})$, where~$H=(V,F,\wFn)$ is a friendship graph,~$\parts\in\N$ is the number of coalitions, and~$\vec{n}$ is a tuple of~$k$ coalition sizes such that~$n_1 \geq n_2 \geq \cdots \geq n_k$ and~$\sum_{i\in[\parts]} n_i = n$.
\end{definition}

The solutions concepts adopted by \citeauthor{BiloMM2022} are those of stability and utilitarian social welfare. However, it is easy to see that any notion of fairness introduced in this work for fair partitioning of friends can be considered, without any modification, also in the context of ASHG-FSC. Moreover, we can observe that, for every fairness notion~$\phi$, fair partitioning of friends is a special case of ASHG-FSC.

\begin{observation}
    Let~$\phi\in\{\text{PROP},\text{MMS},\text{EF},\text{EFX}_0,\text{EFX},\text{EF1}\}$ be a desired fairness notion and~$\mathcal{I}=(G,(\util_a)_{a\in\agents},\parts)$ be an instance of fair partitioning of friends with additive utilities. Then,~$\mathcal{I}$ admits a~$\phi$-fair balanced~$\parts$-partition~$\pi$ if and only if ASHG-FSC~$\Gamma = (H=(A,E,\wFn),\parts,(n_i)_{i\in[\parts]})$, where~$\wFn(\{a,b\}) = \util_a(b)$ and~$n_i = \lfloor n/\parts \rfloor$ if~$i \leq (n\bmod k) + 1$ and~$n_i = \lceil n/\parts \rceil$ otherwise, admits~$\phi$-fair partition~$P$.
\end{observation}

Consequently, we obtain that our model of fair partitioning of friends is a special case of ASHG-FSC. Therefore, whenever it is computationally hard to decide whether a~$\phi$-fair balanced~$\parts$-partition, the same hardness also applies to the decision problem asking whether a~$\phi$-fair coalitions structure exists in ASHG-FSC.

Therefore, a more interesting question is whether any of our positive algorithmic results carry over to the more general model of ASHG-FSC. Maybe surprisingly, it turns out that our algorithms are very robust in this regard. Specifically, our algorithms work even in the setting where the parts are not necessarily balanced, but the sizes are given as part of the input. In other words, they can decide the existence (and find a corresponding partition, if one exists) of~$\phi$-fair partitions also for instances of ASHG-FSC. In the remainder of this section, we describe (very minor) modifications that one needs to perform to make our algorithms work in the not necessarily balanced regime. We focus only on our most advanced tractability results.

We start with the extension of the algorithm from \Cref{thm:binaryFPTbyVC}. This algorithm assumes binary valuations and is \FPT with respect to the vertex cover number of the friendship graph. To extend the algorithm to parts of arbitrary sizes, we add, for each~$i\in[\parts]$, a single constraint that ensures that the size of part~$i$ is not greater than~$n_i$. This modification does not increase the number of variables and does not require extra guessing, so the running time remains the same.

\begin{theorem}
    Let~$\Gamma = (H,\parts,\vec{n})$ be an instance of ASHG-FSC with all edges of weight one. Then, for every notion of fairness~$\phi\in\{\text{PROP},\text{MMS},\text{EF},\text{EFX}_0,\text{EFX},\text{EF1}\}$, there is an \FPT algorithm parameterized by the vertex cover number~$\vc$ of~$H$ that decides whether a~$\phi$-fair partition exists.
\end{theorem}

Next, one can observe that the definitions of signatures and partial partitions in the dynamic programming algorithm for bounded treewidth friendships graphs from \Cref{thm:tw:XP} do not contain any requirement on balanced parts. Consequently, the dynamic programming table stores information about \emph{all}~$\phi$-fair partitions and we can decide the existence of a balanced one just by reading the correct cell of the table in the root of the tree decomposition. Therefore, the only modification we need to accommodate arbitrary sizes vector~$\vec{n}$ is the final step of the algorithm where we, instead of checking the value stored in~$\operatorname{DP}_r$ for a balanced vector~$\vec{b}$, we check for the given vector~$\vec{n}$. The rest of the computation is still the same, and thus the running time remains the same.

\begin{theorem}
    Let~$\Gamma = (H,\parts,\vec{n})$ be an instance of ASHG-FSC with all edges of weight one. Then, for every notion of fairness~$\phi\in\{\text{PROP},\text{MMS},\text{EF},\text{EFX}_0,\text{EFX},\text{EF1}\}$, there is an \XP algorithm parameterized by the treewidth~$\tw$ of~$H$ and the number of parts~$\parts$, combined, that decides whether there is a~$\phi$-fair partition.
\end{theorem}

Similarly, a trivial tweak is required for our result from \Cref{thm:monotone:MMS:exists}. Recall that the algorithm is initialized with a budget vector~$\vec{b} = (\lfloor n/\parts \rfloor, \ldots, \lfloor n/\parts \rfloor, \lceil n/\parts \rceil, \ldots, \lceil n/\parts \rceil)$. However, if we replace balanced~$\vec{b}$ with arbitrary~$\vec{n}$, the correctness arguments and the running time analysis remain the same. This is because the algorithm maintains two invariants independent of the balancedness of the constructed parts. This yields to our final result.

\begin{theorem}
    Let~$\Gamma = (H,\parts,\vec{n})$ be an instance of ASHG-FSC. If the graph~$H$ is a forest, an MMS and EFX partition is guaranteed to exist and can be found in polynomial time.
\end{theorem}

\section{Conclusion}\label{sec:conclusion}

Our results indicate that the problem of fair partitioning of friends heavily depends on the friendship graph and the fairness concept we adopt. Although we have provided a rather complete picture of the complexity of the problem, we believe that there are several directions for future work. Firstly, are there any other graph classes that allow for the tractability of the problem? Our preliminary results indicate that grid-graphs are such a class. What if friendship graphs are directed? Some of our results apply here, but there exist unsettled questions that deserve to be studied. Does the presence of enemies, i.e., agents giving negative utilities, drastically change the landscape of the problem?

A different direction would be to restrict the valuations of the agents and study whether this can lead to tractability. We strongly believe that our proofs of \Cref{thm:EFX:NPh,thm:MMS:nopolytime,thm:EF1_NPh_vc2} can be extended to graphs with constant vertex cover number and valuations given in unary. On the other hand though, we cannot see if the problems remain hard when agents have binary valuations, with or without constant vertex cover number. Are there any polynomial-time algorithms for the problems under binary valuations?

Finally, we should highlight that there exist other fairness notions, even more relaxed than ours, whose existence remains unresolved: PROP1 and EF$r$, for~$r \geq 2$, are among the best candidates to prove some further positive results.

\section*{Acknowledgements}

This work was co-funded by the European Union under the project Robotics and advanced industrial production (reg. no. CZ.02.01.01/00/22\_008/0004590).
Argyrios Deligkas and Stavros D. Ioannidis acknowledge the support of the EPSRC grant EP/X039862/1.
Dušan Knop and Šimon Schierreich acknowledge the support of the Czech Science Foundation through the project No. 22-19557S.
Šimon Schierreich acknowledges the additional support of the Grant Agency of the Czech Technical University in Prague, grant No. SGS23/205/OHK3/3T/18.

\bibliographystyle{plainnat}
\bibliography{references}

\end{document}